\documentclass[10pt,reqno]{amsart}
\usepackage{amsmath,amssymb,amsthm,graphicx,epstopdf,mathrsfs,url}
\usepackage[usenames,dvipsnames]{color}
\usepackage{dsfont}   
\definecolor{darkred}{rgb}{0.7,0.1,0.1}
\definecolor{darkblue}{rgb}{0.1,0.1,0.4}
\definecolor{darkgrey}{rgb}{0.5,0.5,0.5}
\usepackage[colorlinks=true,linkcolor=darkred,citecolor=Blue]{hyperref}

\newtheorem{theorem}{Theorem}[section]

\newtheorem{definition}[theorem]{Definition}

\newtheorem{lemma}[theorem]{Lemma}

\newtheorem{proposition}[theorem]{Proposition}
\newtheorem{remark}[theorem]{Remark}

\numberwithin{equation}{section}

\def\B{\mathsf B}
\def\RE{\mathbb R}
\def\CO{{\mathbb C}}
\def\ee{\eta}
\def\es{\tau}

\newcommand{\hyp}[1]{$C^{2}$-hypersurface as in Definition~\ref{definition_hypersurface}}

\DeclareMathOperator\ran{ran}

\newcommand{\dom}{\mathrm{dom}\,}

\begin{document}
\title[]{Limiting absorption principle and scattering matrix for Dirac operators with $\delta$-shell interactions}
\author[]{}

\author[J. Behrndt]{Jussi Behrndt}
\address{Institut f\"{u}r Angewandte Mathematik\\
Technische Universit\"{a}t Graz\\
 Steyrergasse 30, 8010 Graz, Austria\\
E-mail: {\tt behrndt@tugraz.at}}

\author[M. Holzmann]{Markus Holzmann}
\address{Institut f\"{u}r Angewandte Mathematik\\
Technische Universit\"{a}t Graz\\
 Steyrergasse 30, 8010 Graz, Austria\\
E-mail: {\tt holzmann@math.tugraz.at}}

\author[A.~Mantile]{Andrea Mantile}
\address{
Laboratoire de Math\'ematiques \\
Universit\'e de Reims \\
FR3399 CNRS, Moulin de la Housse, BP 1039, 51687 Reims, France \\
E-mail: {\tt andrea.mantile@univ-reims.fr}
}

\author[A.~Posilicano]{Andrea Posilicano}
\address{
DiSAT, Sezione di Matematica \\
Universit\`a dell' Insubria \\
via Valleggio 11, I-22100 Como, Italy \\
E-mail: {\tt andrea.posilicano@uninsubria.it}
}

\begin{abstract}
  We provide a limiting absorption principle for self-adjoint realizations of
Dirac operators with electrostatic and Lorentz scalar $\delta$-shell
interactions supported on regular compact surfaces. Then we show 
completeness of the wave operators and give a representation formula for the
scattering matrix.
\end{abstract}


\maketitle

\section{Introduction}

The Dirac operator is one of the main mathematical objects in relativistic quantum mechanics. 
Knowledge of its spectral properties leads to the understanding of the behavior of spin-$\frac{1}{2}$ particles like electrons 
in the corresponding physical system. Moreover, the Dirac operator and its spectral properties play  
an important role in the analysis of graphene type materials. 

Since the spectral analysis of Dirac operators with strongly localized potentials is a challenging problem, such potentials are often 
replaced in mathematical physics by singular $\delta$-type potentials. This idea was successfully applied in nonrelativistic quantum mechanics, 
see, e.g., \cite{AGHH, BEHL17-2, BLL13, BEKS94, E08, EH15, MaPo} and the references therein, and in the recent years also in the relativistic setting. 
In this paper we study singular perturbations of the free Dirac operator $A_0$ acting in 
$L^{2}(\mathbb{R}^{3};{\mathbb{C}}^{4})\cong L^{2}(\mathbb{R}^{3})^{4}$, which are formally given by
\begin{equation*}%
\begin{array}
[c]{ccc}%
A_{\eta,\tau}=A_{0}+(\eta I_{4}+\tau\beta)\delta_{\Gamma}\,, &  & \beta:=%
\begin{pmatrix}
I_{2} & 0\\
0 & -I_{2}%
\end{pmatrix}
\,;
\end{array}
\end{equation*}
see Section~\ref{section_free_op} and Section~\ref{section_delta_op} below for the precise definition and the main properties of the appearing objects.
Here $I_{n}$ denotes the identity in $\mathbb{C}^{n,n}$, and  $\delta
_{\Gamma}$ is the tempered distribution supported on the closed bounded
$C^2$-surface $\Gamma$ and acting on a test function $\varphi$ as $\delta_{\Gamma
}(\varphi):=\int_{\Gamma}\varphi(x)\,d\sigma(x)$. The two $\delta
$-perturbation terms with strengths $\eta,\tau\in\mathbb{R}$ 
define the electrostatic shell interaction $\eta I_{4}\delta_{\Gamma}$ and the Lorentz scalar shell interaction
$\tau\beta\delta_{\Gamma}$, respectively.

Singular
perturbations of the Dirac operator have been introduced first in \cite{GeSe}, where the one dimensional Dirac operator with
point interactions is considered, see also \cite{AGHH, BMP17, CMP13, PR14, S89} for more results on Dirac operators with point interactions in $\mathbb{R}$. Shell
interactions supported on a sphere in $\mathbb{R}^3$ were then introduced in \cite{DiExSe} by
using the one-dimensional results and a decomposition to spherical harmonics.
This problem has been recently reconsidered in \cite{ArMaVe 1,AMV15,ArMaVe 3},
where in the case of a $C^2$-surface the self-adjointness and
several properties of Dirac operators with electrostatic $\delta$-perturbations are
derived. An alternative construction of Dirac operators with electrostatic and
Lorentz scalar $\delta $-shell interactions was proposed in
\cite{BEHL17} and further developed in \cite{BEHL19,BeHo}. This approach is based on the method of
quasi boundary triples, originally introduced in \cite{BeLa} for the study of
elliptic partial differential operators. Quasi boundary triples allow to
define distributional perturbations supported on subsets of zero measure, or
more general singular perturbations, as extensions of a symmetric restriction
of an unperturbed operator. This approach easily adapts to the case of Dirac
operators since, in contrast to form methods, no semi-boundedness is required; alternatively one could use the method of self-adjoint extensions of restrictions developed in \cite{P01, P08}.
Next, the fundamental spectral properties of $A_{\eta, \tau}$ under various assumptions on the parameters $\eta$ and $\tau$ were studied in 
\cite{BEHL19, HOP18, MOP18, OP}, see also \cite{BHOP19, PV} for results in the two-dimensional case, and the usage as a model for Dirac operators with strongly localized potentials is justified in some situations in
\cite{MP18} by an approximation result. 
It is also worth mentioning that, modelling $\delta $-shell interactions for
the Dirac operator, a relevant role is played by the parameter $\eta^{2}
-\tau^{2}$; depending on the critical condition $\eta^{2}-\tau^{2}=4$ (so
fixed by our choice of physical units), unexpected spectral effects arise.
While the works mentioned before consider the non-critical case $\eta^{2}
-\tau^{2}\neq4$, the critical regime has been recently investigated in
\cite{BeHo, OV17} and also in \cite{BHOP19}. 

While, as  mentioned above, the spectral properties of $A_{\eta, \tau}$ were investigated, there are hardly no results on scattering theory.
Only  the existence and completeness of the wave operators was shown in the case of electrostatic $\delta $-shell
interactions ($\tau=0$) in \cite{BEHL17} under $C^\infty$-smoothness assumptions on the surface
$\Gamma$; this result was extended in \cite[Proposition~4.7]{BEHL19} for combinations of electrostatic and scalar potentials. 
For this reason, we are concerned in this work with the direct scattering problem for the
couple $(A_{\eta,\tau},A_{0})$. As in most of the above mentioned papers, we consider the three dimensional case; 
nevertheless, using the results from the recent paper \cite{BHOP19} we expect that our approach should also work in space dimension two. In the present paper, we prove  completeness for the scattering
couple $(A_{\eta,\tau},A_{0})$ and provide a representation formula for the
corresponding scattering matrix. More precisely, it will be shown that the
wave operators
\[
W_{\pm}(A_{\eta,\tau},A_{0}):=\text{s-}\lim_{t\rightarrow\pm\infty
}e^{itA_{\eta,\tau}}e^{-itA_{0}}
\]
exist in $L^{2}(\mathbb{R}^{3})^{4}$ and that their ranges 
coincide with the
absolutely continuous subspace of the perturbed operator $A_{\eta,\tau}$. 
Our method to prove
completeness of the wave operators (borrowed from \cite{MaPo}, see Theorem 2.8
there) requires estimates which follow from the limiting absorption principle. Thus our first
goal (and our first main result) in the present paper is to provide a limiting absorption principle for
$A_{\eta,\tau}$ in Theorem~\ref{theorem_LAP}. Due to the lack of semiboundedness this property does not follow directly from the general results in \cite{MaPo}. In this paper we prove the limiting absorption principle by exploiting, besides the limiting absorption principle for
$A_{0}$ and Kre\u{\i}n's resolvent formula
\[
(A_{\eta,\tau}-z)^{-1}=(A_{0}-z)^{-1}-G_{z}\Lambda_{z}^{\eta,\tau}G_{\bar{z}}^{\ast}
\]
as in \cite{MaPo}, some specific properties of the family of operators
$\Lambda_{z}^{\eta,\tau}$ provided in \cite{BeHo}. The limit resolvent at $\lambda \in \mathbb{R}$ then turns
out to have the same structure
\[
R_{\lambda}^{\eta,\tau,\pm}=\lim_{\varepsilon\searrow0}(A_{\eta,\tau
}-(\lambda\pm i\epsilon))^{-1}=R_{\lambda}^{0,\pm}-G_{z}^{\pm}\Lambda
_{\lambda}^{\eta,\tau,\pm}G_{\lambda}^{\mp\ast}.
\]
Once existence and completeness for the wave operators is achieved, we can
define the scattering operator $S_{\eta,\tau}:=W_{+}(A_{\eta,\tau}
,A_{0})^{\ast}W_{-}(A_{\eta,\tau},A_{0})$ and (the physically relevant)
scattering matrix $S_{\eta,\tau}(\lambda)$ via
$$S_{\eta,\tau}(\lambda)(F_{0} f)(\lambda)=(F_{0}S_{\eta,\tau} f)(\lambda),$$ 
where $F_{0}$ is the unitary map which diagonalizes the free Dirac
operator $A_{0}$. In our second main result Theorem~\ref{theorem_scattering_matrix} we provide a
representation formula for $S_{\eta,\tau}(\lambda)$ in terms of the limit
operators $\Lambda_{\lambda}^{\eta,\tau,+}$ appearing in the resolvent formula
above. In order to get such a representation, we follow the same scheme as in
\cite[Section 4]{MaPo}: Birman-Yafaev stationary scattering theory for the
resolvent couple $(-R_{\mu}^{\eta,\tau},-R_{\mu}^{0})$ and Kato-Birman
invariance principle. We also refer the reader to \cite{AP86,BMN08,BeMaNa,BH19-2} for a closely related approach to scattering theory in the context of extension methods and 
Kre\u{\i}n's resolvent formula. Moreover, for a comprehensive list of references on the limiting absorption principle for Dirac operators with regular potentials we refer to \cite{Fritz}.

The paper is organized as follows: In Section~\ref{section_preliminaries} we recall the definition of weighted Sobolev spaces,
the limiting absorption principle for the free Dirac operator, and we study some families of operators which are related to the resolvent of 
the free Dirac operator. 
Section~\ref{section_delta_op} focuses on the rigorous definition and the spectral properties of $A_{\eta, \tau}$; here the main result is 
the limiting absorption principle for $A_{\eta, \tau}$. Finally, in Section~\ref{section_scattering} we prove completeness for the scattering
couple $(A_{\eta,\tau},A_{0})$ and provide a  formula for the
 scattering matrix.

\subsection*{Notations}

By $\mathbb{C}_\pm$ we denote the upper and lower complex half plane, respectively. Let $X$ and $Y$ be Hilbert spaces. We use for $n \in \mathbb{N}$ the notation $X^n := X \otimes \mathbb{C}^n$; the elements of $X^n$ are vectors with entries in $X$. Next  $\B(X,Y)$ is the set of all bounded and everywhere defined operators from $X$ to $Y$. The anti-dual operator of $A \in \B(X, Y)$ is denoted by $A^*$ and maps from $Y'$ to $X'$. If $A$ is a closed operator, then $\dom A$ and $\ran A$ denote the domain of definition and the range of $A$, respectively. If $A$ is self-adjoint, then we denote by $\mathsf{res}(A)$, $\sigma(A)$, $\sigma_\text{p}(A)$, $\sigma_\text{disc}(A)$, $\sigma_\text{ess}(A)$, and $\sigma_\text{ac}(A)$ the resolvent set, the spectrum, the point, the discrete, the essential, and the absolutely continuous spectrum of $A$, respectively. For $z \in \mathsf{res}(A)$ we often write $R_z := (A-z)^{-1}$. Finally, for an open set $\Omega \subset \mathbb{R}^3$ the $L^2$-based Sobolev spaces of order $s \in \mathbb{R}$ are denoted by $H^s(\Omega)$, while the Sobolev space on  a sufficiently regular surface $\Gamma$ are denoted by $H^s(\Gamma)$.

\section{Preliminaries} \label{section_preliminaries}

In this section we collect some preliminary material which is needed to 
formulate and prove the limiting absorption principle for Dirac operators with singular interactions in Section~\ref{section_delta_op}. 
We recall the definitions of weighted Sobolev spaces, the free Dirac operator $A_0$, and provide a 
limiting absorption principle for its resolvent. We also discuss some 
auxiliary operators associated to the resolvent of $A_0$ which are crucial to study the Dirac operator $A_{\eta, \tau}$ with a $\delta$-potential.

\subsection{Weighted Sobolev spaces} \label{section_weighted_Sobolev_spaces}

In the formulation of the limiting absorption principle 
weighted $L^2$-spaces $L^2_w(\mathbb{R}^3)$ and weighted Sobolev spaces $H^s_w(\mathbb{R}^3)$ play an important role. The definition of these spaces below follows for 
indices $s \in \mathbb{N}_0$ the classical one in \cite{Ag} and is extended to general $s \in \mathbb{R}$ via interpolation; cf. \cite[page~245]{ST} and also \cite[Appendix~B]{McLe}.

\begin{definition}
\label{Notation 1}
Let $\left\langle x\right\rangle :=(  1+\vert x\vert ^{2})^{1/2}$ and $w \in \mathbb{R}$. Then we define the weighted $L^2$-space by
\begin{equation*} 
L_{w}^{2}(  \mathbb{R}^{3})  :=\bigl\{f\in\mathcal{S}^\prime(\mathbb{R}^{3}) : \langle x\rangle^{w} f\in
L^{2}(\mathbb{R}^{3})\bigr\}  
\end{equation*}
with norm 
\begin{equation*}
  \| f \|_{L^2_w(\mathbb{R}^3)}^2 := \int_{\mathbb{R}^3} ( 1 + |x|^2 )^w |f(x)|^2 \textup{d} x.
\end{equation*}
The weighted Sobolev spaces of order $l \in \mathbb{N}_0$ are defined by 
\begin{equation*}
  H_w^l(\mathbb{R}^3) = \bigl\{ f \in L^2_w(\mathbb{R}^3): D^\alpha f \in L^2_w(\mathbb{R}^3)\, \forall \alpha \in \mathbb{N}_0^3, |\alpha| \leq l \bigr\},
\end{equation*}
where $D^\alpha$ denotes the weak derivative (of order $\alpha\in \mathbb{N}_0^3$), and equipped with the
norms
\begin{equation*}
  \| f \|_{H_w^l(\mathbb{R}^3)}^2 = \sum_{|\alpha| \leq l} \| D^\alpha f \|_{L^2_w(\mathbb{R}^3)}^2.
\end{equation*}
If $t_0 < t_1$ are two natural numbers, $\theta \in (0,1)$, and 
$s=(  1-\theta)  t_{0}+\theta t_{1}$, then we define 
$H^s_w(\mathbb{R}^3)$ (and a Hilbert space norm) via interpolation 
\begin{equation*}
H_{w}^{s}(  \mathbb{R}^{3})  :=\left[  H_{w}^{t_{0}}(\mathbb{R}^{3})  ,H_{w}^{t_{1}}(  \mathbb{R}^{3})  \right]_{\theta},
\end{equation*}
and for $s<0$ we set $H_{w}^{s}(  \mathbb{R}^{3})  := ( H_{-w}^{-s} (\mathbb{R}^3))'$
equipped with the corresponding norm.
\end{definition}


Next, we state several known results on the trace operator which enter in the construction of singular
perturbations of the free Dirac operator. Let $\Omega\subset\mathbb{R}^{3}$ be an open and bounded
$C^2$-domain, i.e. $\Gamma=\partial\Omega$ is a closed bounded
surface of class $C^2$. We denote
\begin{equation*}
\begin{array}
[c]{ccc}
\Omega_{-}=\Omega, &  & \Omega_{+}=\mathbb{R}^{3}\setminus \overline{\Omega}.
\end{array}
\end{equation*} 
The lateral traces on $\Gamma$ are defined on $C^\infty(\overline{\Omega_{\pm}})^4  $ by $\gamma_{0}^{\pm}u_{\pm}:=u_{\pm
}\vert_\Gamma$. These extend to bounded surjective maps $\gamma_{0}^{\pm}\in\mathsf{B}\left(  H^{1/2+s}(
\Omega_{\pm})^4  ,H^{s}(  \Gamma)^4  \right)$, $s\in (0, \frac{3}{2}]$, see, e.g., 
\cite[Theorem~3.37]{McLe}.
The trace on $\Gamma$ is defined as the mean value
\begin{equation*}
\gamma_{0}:=\frac{1}{2} \left(  \gamma_{0}^{+}+\gamma_{0}^{-}\right)  
\end{equation*}
and will be viewed as a bounded operator from either $H^{s+1/2}(\mathbb{R}^{3})^4$ or 
$H^{s+1/2}(  \mathbb{R}^{3}\setminus\Gamma)^4$ to $H^s(\Gamma)^4$ for $s\in (0, \frac32]$; from the context it will be clear on which space $\gamma_0$ is defined. 
Since
$\Gamma$ is a bounded set it is also clear that $\gamma_0$
is bounded as an operator defined on the weighted spaces $H_w^{s+1/2}(  \mathbb{R}^{3})^4$, more precisely, we have 
\begin{equation*}
\gamma_{0}\in\mathsf{B}\bigl(  H_w^{s+1/2}(  \mathbb{R}^{3})^4
,H^{s}(  \Gamma)^4  \bigr)  \,,\ s\in \big(0,\tfrac32\big],\,\,\,w\in\mathbb R, 
\end{equation*}
and for the anti-dual operator it follows
\begin{equation}\label{allesklar}
\gamma_0^*\in\mathsf{B}\bigl(H^{-s}(\Gamma)^4,  H_{-w}^{-s-1/2}(  \mathbb{R}^{3})^4 \bigr)  \,,\ s\in \big(0, \tfrac32\big],\,\,\,w\in\mathbb R;
\end{equation}
here $\gamma_0^*$ is defined by 
$(\gamma_0^*\varphi)(f)=(\varphi,\gamma_0 f)_{H^{-s}(\Gamma)^4\times H^s(\Gamma)^4}$ for $\varphi\in H^{-s}(\Gamma)^4$, $f\in H_w^{s+1/2}(  \mathbb{R}^{3})^4$,
and $(\cdot,\cdot)_{H^{-s}(\Gamma)^4\times H^s(\Gamma)^4}$ denotes the extension of the $L^2$-scalar product to the dual pair $H^{-s}(\Gamma)^4\times H^s(\Gamma)^4$.

\subsection{The limiting absorption principle for the free Dirac operator} \label{section_free_op}

In this section we recall the definition of the free Dirac operator and how the limiting absorption principle for 
its resolvent can be proved. Many of the mapping properties below can be shown in (weighted) Sobolev spaces 
$H^s_w(\mathbb{R}^3)^4$ for any $s \in \mathbb{R}$, but for simplicity we state them just for those $s$ which are needed later in our applications.
Let $\sigma_{j}\in\mathbb{C}^{2,2}$, $j=1,2,3$, denote the Pauli matrices%
\begin{equation*}%
\begin{array}
[c]{ccccc}%
\sigma_{1}=%
\begin{pmatrix}
0 & 1\\
1 & 0
\end{pmatrix}
\,, &  & \sigma_{2}=%
\begin{pmatrix}
0 & -i\\
i & 0
\end{pmatrix}
\,, &  & \sigma_{3}=%
\begin{pmatrix}
1 & 0\\
0 & -1
\end{pmatrix}
\,,
\end{array}
\end{equation*}
and $\alpha_{j}$, $\beta\in\mathbb{C}^{4,4}$, $j=1,2,3$, the Dirac matrices
\begin{equation*}%
\begin{array}
[c]{ccc}%
\alpha_{j}=%
\begin{pmatrix}
0 & \sigma_{j}\\
\sigma_{j} & 0
\end{pmatrix}
\,, &  & \beta=%
\begin{pmatrix}
I_{2} & 0\\
0 & -I_{2}%
\end{pmatrix}
\,,
\end{array}
\end{equation*}
where $I_n$
is the identity in
$\mathbb{C}^{n,n}$. We will often use for $x=(x_1,x_2,x_3) \in \mathbb{R}^3$ the notations $\alpha \cdot x = \alpha_1 x_1 + \alpha_2 x_2 + \alpha_3 x_3$ and $\alpha \cdot \nabla = \alpha_1 \partial_1 + \alpha_2 \partial_2 + \alpha_3 \partial_3$.

In the units $\hbar=c=1$ the Dirac operator $A_{0}$ for a free
relativistic particle of mass $m=1$ is the unbounded self-adjoint operator in
$L^{2}(\mathbb{R}^{3})^4$ defined by
\begin{equation}
A_{0}=-i\sum_{j=1}^{3}\alpha_{j}\partial_{j}+\beta,\quad\dom  
A_{0}  =H^{1}(  \mathbb{R}^{3})^4  . \label{D_op}%
\end{equation}
Its spectrum is
\begin{equation*}
\sigma(A_0)  =\sigma_\textup{ac}(A_0)  =\left(
-\infty,-1\right]  \cup\left[  1,\infty\right)
\end{equation*}
and one has $\sigma_\text{p}(A_0)  =\varnothing$. The operator
$A_{0}$ in \eqref{D_op} can also be viewed as an operator from $H^{1}%
(\mathbb{R}^{3})^4$ to $L^{2}(\mathbb{R}^{3})^4$, where it is also bounded, so
$A_{0}\in\mathsf{B}(  H^{1}(\mathbb{R}^{3})^4,L^{2}(\mathbb{R}^{3})^4)
$. For $z\in\mathsf{res}\left(  A_{0}\right)  $ the operator $A_{0}%
-z\in\mathsf{B}(  H^{1}(\mathbb{R}^{3})^4,L^{2}(\mathbb{R}^{3})^4)  $
is bijective, and by duality one also has that $A_{0}-z\in\mathsf{B}(
L^{2}(\mathbb{R}^{3})^4,H^{-1}(\mathbb{R}^{3})^4)  $ is bijective for
$z\in\mathsf{res}\left(  A_{0}\right)  $. Hence, by interpolation
\begin{equation}
A_{0}-z\in\mathsf{B}\big(  H^{-s+1}(\mathbb{R}^{3})^4,H^{-s}(\mathbb{R}%
^{3})^4\big)  ,\qquad s\in\lbrack0,1],~z\in\mathsf{res}\left(A_{0}\right), \label{azint}%
\end{equation}
is bijective.
Setting $R_{z}^{0}:=\left(  A_{0}-z\right)  ^{-1}$, $z\in\mathsf{res}\left(
A_{0}\right)  $, we obtain the following lemma.

\begin{lemma}
\label{Lemma_Res_map} For $s\in[0,1]$ the map
\begin{equation*}
z\rightarrow R_{z}^{0}\in\mathsf{B}\big(  H^{-s}(  \mathbb{R}^{3})^4  ,H^{-s+1}(  \mathbb{R}^{3})^4  \big) 
\end{equation*}
is holomorphic on $\mathsf{res}(A_0)=\mathbb{C}\setminus 
\big((-\infty,-1]  \cup [  1,\infty) \big)$.
\end{lemma}

\begin{proof}
Fix $s\in[0,1]$ and $z_0\in \mathsf{res}(A_0)$, and let $z\in B_{\delta}(z_{0})$ with $\delta>0$ sufficiently small. From the identity 
\begin{equation*}
R_{z}^{0}=\left(  1 - R_{z_{0}}^{0}\left(  z-z_{0}\right)  \right)
^{-1}R_{z_{0}}^{0}
\end{equation*}
it follows that the family $\{  R_{z}^{0} : z\in B_{\delta}(z_{0})\}$ is uniformly bounded with respect to the norm 
in $\mathsf{B}(H^{-s}(  \mathbb{R}^{3})^4  ,H^{-s+1}(  \mathbb{R}^{3})^4)$. Now the resolvent identity
\begin{equation}
\label{resid}R_{z}^{0}-R_{z_{0}}^{0}=(z-z_{0})R_{z}^{0} R_{z_{0}}^{0}
\end{equation}
implies first that the map $z\rightarrow R_{z}^{0}$ is continuous in $z_{0}\in\mathsf{res}(A_0)$ with values in 
$\mathsf{B}(H^{-s}(\mathbb{R}^{3})^4 ,H^{-s+1}(\mathbb{R}^{3})^4)$. In a second step \eqref{resid} implies
that $z\rightarrow R_{z}^{0}$ is holomorphic with values in $\mathsf{B}(H^{-s}(\mathbb{R}^{3})^4 ,H^{-s+1}(\mathbb{R}^{3})^4)$.
\end{proof}

It is not difficult to check that the Dirac operator in \eqref{D_op} is
bounded as an operator from $H_{w}^{l+1}(\mathbb{R}^{3})^4$ to $H_{w}^{l}(\mathbb{R}^{3})^4$ for $l\in \mathbb{N}_0$
and $w\in\mathbb{R}$, in particular, 
\begin{equation*}
A_{0}\in\mathsf{B}\big(  H_w^{l+1}(\mathbb{R}^{3})^4,H_{w}^{l}(\mathbb{R}^{3})^4\big),\qquad w\in\mathbb{R},\,\, l=0,1.
\end{equation*}
By duality, one has $A_{0}\in\mathsf{B}(  H_{-w}^{-l}(\mathbb{R}^{3})^4,H_{-w}^{-l-1}(\mathbb{R}^{3})^4)$ for $l=0,1$
and $w\in\mathbb{R}$, and hence $A_{0}\in\mathsf{B}(  H_{w}^{-l}(\mathbb{R}^{3})^4,H_{w}^{-l-1}(\mathbb{R}^{3})^4)$ for $l=0,1$
and $w\in\mathbb{R}$. Interpolation yields
\begin{equation}
A_{0}\in\mathsf{B}\big(  H_{w}^{-s+1}(\mathbb{R}^{3})^4,H_{w}^{-s}(\mathbb{R}^{3})^4\big)  ,\qquad w\in\mathbb{R},\,\,s\in[-1,1],
\label{aow}%
\end{equation}
in analogy with \eqref{azint}. This property extends to all $s\in\mathbb R$, but only $s\in[-1,1]$ is needed here.

%

Next we provide some properties of the resolvent of $A_{0}$
and its limit behaviour when $z$ tends from $\mathbb{C}_\pm$ to the continuous spectrum. In particular, it
turns out that the resolvent $z\mapsto R_{z}^{0}$ extends continuously to
$\lambda\pm i0$, $\lambda\in (  -\infty,-1)\cup(1,\infty)$, in the weaker topology of 
$\mathsf{B}(  H_{w}^{-s}(  \mathbb{R}^{3})^4  ,H_{-w}^{-s+1}(  \mathbb{R}^{3})^4 )  $ for $w>1/2$.

\begin{proposition}
\label{Proposition_Res_D} The resolvent $R_{z}^{0}$ of the free Dirac operator
$A_{0}$ in \eqref{D_op} has the following properties.

\begin{itemize}
\item[$(i)$] For $w\in\mathbb{R}$ and $s\in [0,1]$
we have $R_{z}^{0}\in\mathsf{B}(H_{w}^{-s}(\mathbb{R}^{3})^4,H_{w}^{-s+1}(\mathbb{R}^{3})^4)$, $z\in\mathsf{res}(A_0)$.

\item[$(ii)$] For $w>1/2$ and $s\in [0,1]$ the limits
\begin{equation*}
R_{\lambda}^{0,\pm}:=\lim_{\varepsilon\searrow0}R_{\lambda\pm i\varepsilon}^{0}\,,\quad\lambda\in\left(  -\infty,-1\right) \cup \left(  1,\infty\right)  ,
\end{equation*}
exist in $\mathsf{B}(H_{w}^{-s}(\mathbb{R}^{3})^4, H_{-w}^{-s+1}(\mathbb{R}^{3})^4 )$ and the maps
\begin{equation} \label{R_pm_continuous}
  z \mapsto R_z^{0, \pm} := \begin{cases} R_z^0, & z \in \mathbb{C} \setminus ( (-\infty, -1] \cup [1, \infty) ), \\ R_\lambda^{0, \pm}, & z=\lambda \in (-\infty, -1) \cup (1, \infty), \end{cases}
\end{equation}
are continuous from $\overline{\mathbb{C}_\pm} \setminus \{ -1, 1 \}$ to $\mathsf{B}(  H_{w}^{-s}( \mathbb{R}^{3})^4, H_{-w}^{-s+1}(\mathbb{R}^{3})^4)$.
Moreover, each limit
$R_{\lambda}^{0,\pm}$ defines a right inverse of $\left(  A_{0}-\lambda \right)  $, i.e.%
\begin{equation*}
\left(  A_{0}-\lambda\right)  R_{\lambda}^{0,\pm}= I_4.
\end{equation*}

\end{itemize}
\end{proposition}

The proof of Proposition~\ref{Proposition_Res_D} below is making use of the
mapping properties of the resolvent of the Laplacian. More precisely, let
$-\Delta$ denote the self-adjoint Laplace operator in $L^{2}(\mathbb{R}^{3})$
defined on $H^{2}(\mathbb{R}^{3})$ and
\begin{equation*}
r_{z}^{0}:=\left(  -\Delta-z\right)^{-1},\qquad z\in\mathbb{C}\setminus [0, \infty).
\end{equation*}
We first recall some known mapping properties of $r_{z}^{0}$.

\begin{lemma}
\label{Lemma_rz0} The resolvent $r_{z}^{0}$ of the free Laplacian $-\Delta$
has the following properties.

\begin{itemize}
\item[$(i)$] For $w\in\mathbb{R}$ and $s\in\left[  0,2\right]$ we have $r_{z}^{0}\in\mathsf{B}(H_{w}^{-s}(\mathbb{R}^{3})  ,H_{w}^{-s+2}(\mathbb{R}^{3}) )$, 
$z\in\mathbb{C}\setminus [0, \infty)$.

\item[$(ii)$] For $w>1/2$ and $s\in\lbrack0,2]$ the limits
\begin{equation*}
r_{\lambda}^{0,\pm}:=\lim_{\varepsilon\searrow0}r_{
\lambda\pm i\varepsilon}^{0}\,,\quad\lambda>0,
\end{equation*}
exist in $\mathsf{B}(  H_{w}^{-s}(\mathbb{R}^{3}), H_{-w}^{-s+2}(\mathbb{R}^{3}) )$ and the maps
\begin{equation*}
  z \mapsto r_z^{0, \pm} := \begin{cases} r_z^0, & z \in \mathbb{C} \setminus [0, \infty), \\ r_\lambda^{0, \pm}, & z=\lambda \in (0, \infty), \end{cases}
\end{equation*}
are continuous from $\overline{\mathbb{C}_\pm} \setminus \{ 0 \}$ to $\mathsf{B}(H_{w}^{-s}( \mathbb{R}^{3}), H_{-w}^{-s+2}(\mathbb{R}^{3}))$. Moreover, each limit
$r_{\lambda}^{0,\pm}$ defines a right inverse of $\left(  -\Delta
-\lambda\right)  $, i.e.%
\begin{equation*}
\left(  -\Delta-\lambda\right)  r_{\lambda}^{0,\pm}= I_1.
\end{equation*}
\end{itemize}
\end{lemma}

\begin{proof}
$(i)$ By \cite[equation~(4.8)]{MaPoSi18} we have
\begin{equation}\label{4.8}
r_{z}^{0}\in\mathsf{B}\left(  L_{w}^{2}(\mathbb{R}^{3}), H_{w}^{2}(\mathbb{R}^{3})  \right),\qquad w\in\mathbb R,\,\,\,z\in\mathbb{C}\setminus [0, \infty),
\end{equation}
(alternatively, \eqref{4.8} can be proved starting from the obvious unweighted estimate 
$\|f\|_{{H^{2}(\RE^{3})}}\le C\,\|(-\Delta+z)f\|_{{L^{2}(\RE^{3})}}$ and then passing to the 
weighted one by using \cite[estimate (A.17)]{Ag}). Thus, by
duality we conclude $r_{z}^{0}\in\mathsf{B}(H_{w}^{-2}(\mathbb{R}^{3}), L_{w}^{2}(\mathbb{R}^{3}))$ and hence,
by interpolation $r_{z}^{0}\in\mathsf{B}(H_{w}^{-s}(\mathbb{R}^{3}), H_{w}^{-s+2}(\mathbb{R}^{3}))$
for all $s\in [0,2]$, $w\in\mathbb R$, and $z\in\mathbb{C}\setminus [0, \infty)$.

Assertion $(ii)$ can be shown in the same way as item~$(i)$ using \cite[Theorem~4.1]{Ag}, see also \cite[Theorem~18.3]{KK}, for $s=0$, duality for $s=-2$, 
and an interpolation argument for $s \in (-2, 0)$. The main ingredient in the quoted theorem is the weighted inequality $\|f\|_{H^{2}_{-w}(\RE^{3})}\le C\,\|(-\Delta+z)f\|_{L^{2}_{w}(\RE^{3})}$, which holds  for any $f\in H^{2}(\RE^{3})$, $z\in\CO$, $|z|\in [K^{-1},K]$, and a constant $C$ depending only on $w$ and $K$, whenever $w>\frac12$ and $K>1$ (see \cite[Lemma 4.1]{Ag}).
\end{proof}

\begin{proof}[Proof of Proposition \ref{Proposition_Res_D}]
For $z\in\mathsf{res}\left(
A_{0}\right)  $ we make use of the identity (see, e.g. \cite[eq. (1.3)]{BaHe})
\begin{equation*}
\left(  A_{0}-z\right)  \left(  A_{0}+z\right)  =(-\Delta+1-z^{2})  I_{4},
\end{equation*}
which leads to
\begin{equation*}
R_{z}^{0}=\left(  A_{0}+z\right)  r_{\left(  z^{2}-1\right)  }^{0}I_{4}.
\end{equation*}
Note that $z\in(-\infty,-1)\cup(1,\infty)$ if and only if $(z^{2}-1)>0$.
Now assertions $(i)$-$(ii)$ follow from items $(i)$-$(ii)$ in
Lemma~\ref{Lemma_rz0} and \eqref{aow}.
\end{proof}

Finally, we consider the symmetric restriction $S$ of $A_0$ to $H^{1}_0(\mathbb{R}^{3}\setminus\Gamma)^4$, that is, 
\begin{equation*}
S=-i\sum_{j=1}^{3}\alpha_{j}\partial_{j}+\beta,\quad\dom  S  =\bigl\{f\in H^{1}(\mathbb{R}^{3})^4:\gamma_0 f=0\bigr\}.
\end{equation*}
In Section~\ref{section_delta_op} we define Dirac operators $A_{\eta, \tau}$ with $\delta$-interactions as self-adjoint extensions of $S$. It can be 
shown that the adjoint $S^*$ has the form 
\begin{equation} \label{def_S_star}
  \begin{split}
    \dom   S^* &= \big\{ f = f_+ \oplus f_- \in L^2(\Omega_+)^4 \oplus L^2(\Omega_-)^4: \alpha \cdot \nabla f_\pm \in L^2(\Omega_\pm)^4 \big\}, \\
    S^* f &= (-i \alpha \cdot \nabla + \beta) f_+ \oplus (-i \alpha \cdot \nabla + \beta) f_-,
  \end{split}
\end{equation}
where the derivatives are understood in the distributional sense, cf. \cite[Proposition~3.1]{BeHo}.
In the next lemma we recall a result on the extension of the trace maps $\gamma_0^\pm$ from \cite[Proposition~2.1]{OV17}, see also \cite[Lemma~4.3]{BeHo}. In the formulation of the result we use for a function $f \in L^2(\mathbb{R}^3)^4$ the notation $f_\pm = f \upharpoonright \Omega_\pm$.

\begin{lemma}\label{Lemma_trace}
The trace map 
\begin{equation*}
  \gamma_0^\pm: H^1(\mathbb{R}^3 \setminus \Gamma)^4 = H^1(\Omega_+)^4 \oplus H^1(\Omega_-)^4 \rightarrow H^{1/2}(\Gamma)^4, 
  \quad \gamma_0^\pm (f_+ \oplus f_-) = f_\pm|_\Gamma,
\end{equation*}
extends by continuity to 
\begin{equation*}
  \gamma_{0}^\pm\in\mathsf{B}\bigl(\dom   S^{\ast}, H^{-1/2}(\Gamma)^4  \bigr),
\end{equation*}
where $\dom   S^*$ is equipped with the graph norm of $S^*$.
\end{lemma}

\subsection{Auxiliary maps and estimates} \label{section_G_z}

In this section we study the operator functions $G_z$ and $M_z$ given by
\begin{equation}\label{gz}
G_z=R_{z}^{0}\gamma_0^*\quad \text{and}\quad
M_{z}=\gamma_{0}G_{z},\qquad z\in\mathsf{res}(A_0).
\end{equation}
These operators play a crucial role in our construction in the next section. In what 
follows, we discuss their mapping properties and their limit behaviour, when the spectral parameter $z \in \mathbb{C}_\pm$ approaches the continuous spectrum.

\begin{proposition}\label{Proposition_G_Z_map}
For the operators $G_z$ in \eqref{gz} the following is true.
\begin{itemize}
\item[$(i)$] For all $z_1, z_2 \in\mathsf{res}(A_0)$
\begin{equation*}
  G_{z_1} - G_{z_2} = (z_1 - z_2) R^0_{z_1} G_{z_2} = (z_1 - z_2) R^0_{z_2} G_{z_1}
\end{equation*}
holds.
\item[$(ii)$] The map $z\rightarrow G_{z}\in\mathsf{B}(H^{-1/2}(\Gamma)^4,  L^2(\mathbb{R}^{3})^4)$ is holomorphic
on $\mathsf{res}(  A_{0})$.
\item[$(iii)$] For $w>1/2$ the limits 
\begin{equation}\label{limgz123}
G_\lambda^\pm:=\lim_{\varepsilon\searrow 0}G_{\lambda \pm i \varepsilon},\qquad \lambda\in(  -\infty,-1)  \cup(  1,\infty),
\end{equation}
exist in $\mathsf{B}(H^{-1/2}(\Gamma)^4,  L^2_{-w}(\mathbb{R}^{3})^4)$, one has 
$$G_\lambda^\pm=R_{\lambda}^{0,\pm}\gamma_0^* \in \mathsf{B}\bigl(H^{-1/2}(\Gamma)^4,  L^2_{-w}(\mathbb{R}^{3})^4\bigr),$$ 
and the maps
\begin{equation} \label{G_pm_continuous}
  z \mapsto G_z^{\pm} := \begin{cases} G_z, & z \in \mathbb{C} \setminus ( (-\infty, -1] \cup [1, \infty) ), \\ G_\lambda^{\pm}, & z=\lambda \in (-\infty, -1) \cup (1, \infty), \end{cases}
\end{equation}
are continuous from $\overline{\mathbb{C}_\pm} \setminus \{ -1, 1 \}$ to $\mathsf{B}(H^{-1/2}(\Gamma)^4,  L^2_{-w}(\mathbb{R}^{3})^4)$. 
\item[$(iv)$] For any compact  ${I} \subset \mathbb{R} \setminus \{ -1, 1 \}$
\begin{equation}\label{sup_G_z}
\sup_{(\lambda,\varepsilon)\in I\times (0,1)}\sqrt\varepsilon\,\|G_{\lambda\pm i\varepsilon}\|_{H^{-1/2}(\Gamma)^4, L^{2}(\mathbb{R}^3)^4}<\infty
\end{equation}
holds.
\item[$(v)$] The dual $G_z^* \in \mathsf{B}(L^2(\mathbb{R}^{3})^4, H^{1/2}(\Gamma)^4)$ of $G_z$ is  given by
\begin{equation*}
  G_z^*: L^2(\mathbb{R}^3)^4 \rightarrow H^{1/2}(\Gamma)^4, \qquad G_z^* f = \gamma_0 R^0_{\Bar z} f,
\end{equation*}
and the map $\mathsf{res}(  A_{0}) \ni z \mapsto G_{\Bar z}^*$ is holomorphic in $\mathsf{B}(L^2(\mathbb{R}^{3})^4, H^{1/2}(\Gamma)^4)$.
\item[$(vi)$] For $w>1/2$ the limits 
\begin{equation} \label{lim_G_star}
(G_\lambda^\pm)^*:=\lim_{\varepsilon \searrow 0}(G_{\lambda \pm i \varepsilon})^*,\qquad \lambda\in(  -\infty,-1)  \cup(  1,\infty),
\end{equation}
exist in $\mathsf{B}(L^2_{w}(\mathbb{R}^{3})^4, H^{1/2}(\Gamma)^4)$,  one has
$$(G_\lambda^\pm)^*=\gamma_0 R_{\lambda}^{0,\mp} \in \mathsf{B}\bigl(L^2_{w}(\mathbb{R}^{3})^4, H^{1/2}(\Gamma)^4\bigr),$$ 
and the maps
\begin{equation*}
  z \mapsto (G_z^{\pm})^* := \begin{cases} G_z^*, & z \in \mathbb{C} \setminus ( (-\infty, -1] \cup [1, \infty) ), \\ (G_\lambda^{\pm})^*, & z=\lambda \in (-\infty, -1) \cup (1, \infty), \end{cases}
\end{equation*}
are continuous from $\overline{\mathbb{C}_\mp} \setminus \{ -1, 1 \}$ to $\mathsf{B}(L^2_{w}(\mathbb{R}^{3})^4, H^{1/2}(\Gamma)^4)$.
\end{itemize}
\end{proposition}

\begin{proof}
Item $(i)$ is a simple consequence of the definition of $G_z$ in~\eqref{gz} and the resolvent identity.

$(ii)$ By Lemma~\ref{Lemma_Res_map} applied for $s=1$ the map $z\mapsto R_z^0\in\mathsf{B}(H^{-1}(\mathbb R^3)^4, L^2(\mathbb R^3)^4)$ is holomorphic. 
Together with \eqref{allesklar} for $w=0$ and $s=1/2$ we conclude $(ii)$.

$(iii)$ For $w>1/2$, $s\in [0,1]$, and $\lambda\in(-\infty,-1)\cup(1,\infty)$ 
the limits $R_\lambda^{0,\pm}=\lim_{\varepsilon \searrow 0}R_{\lambda \pm i \varepsilon}$
exist in $\mathsf{B}(H_w^{-1}(\mathbb R^3)^4, L^2_{-w}(\mathbb R^3)^4)$ according to 
Proposition~\ref{Proposition_Res_D}~$(ii)$, again applied with $s=1$. From \eqref{allesklar} with $s=1/2$ we conclude that the limits 
\begin{equation*}
G_\lambda^\pm=
\lim_{\varepsilon \searrow 0}G_{\lambda \pm i \varepsilon}=
\lim_{\varepsilon \searrow 0}R_{\lambda \pm i \varepsilon}^0\gamma_0^*
\end{equation*}
exist in $\mathsf{B}(H^{-1/2}(\Gamma)^4, L^2_{-w}(\mathbb R^3)^4)$ and one has
\begin{equation*}
 G_\lambda^\pm=R_\lambda^{0,\pm}\gamma_0^* \in\mathsf{B}\bigl(H^{-1/2}(\Gamma)^4,  L^2_{-w}(  \mathbb{R}^{3})^4 \bigr),\qquad w>1/2.
\end{equation*}
Therefore, the continuity in~\eqref{G_pm_continuous} is a simple consequence of~\eqref{R_pm_continuous} for $s=1$. 

$(iv)$ The claim is a consequence of the limiting absorption principle for $G_z$. It follows from the estimate 
(3.16) in \cite{MaPo} and \eqref{R_pm_continuous}.

$(v)$-$ (vi)$ The claims follow directly from $(ii)$ and $(iii)$ by duality.
\end{proof}

Next, we discuss the operators $M_z$ which are formally given by~\eqref{gz}. 

\begin{proposition}\label{Proposition_m_Z_map}
For the operators $M_z$ in \eqref{gz} the following is true.
\begin{itemize}
\item[$(i)$] For all $z\in \mathsf{res}(  A_{0})$ one has  $M_z \in \mathsf{B}(H^{-1/2}(\Gamma)^4)$.
\item[$(ii)$] For all $z_1, z_2 \in \mathsf{res}(  A_{0})$
\begin{equation*}
  M_{z_1} - M_{z_2} = (z_1 - z_2) G_{\Bar z_1}^* G_{z_2} = (z_1 - z_2) G_{\Bar z_2}^* G_{z_1}
\end{equation*}
holds.
\item[$(iii)$] The map $z\rightarrow M_{z}\in\mathsf{B}(H^{-1/2}(\Gamma)^4)$ is holomorphic
on $\mathsf{res}(  A_{0})$.
\item[$(iv)$] The limits 
\begin{equation}\label{limmz123}
M_\lambda^\pm:=\lim_{\varepsilon\searrow 0}M_{\lambda \pm i \varepsilon},\qquad \lambda\in(  -\infty,-1)  \cup(  1,\infty),
\end{equation}
exist in $\mathsf{B}(H^{-1/2}(\Gamma)^4)$ and the maps
\begin{equation*} 
  z \mapsto M_z^{\pm} := \begin{cases} M_z, & z \in \mathbb{C} \setminus ( (-\infty, -1] \cup [1, \infty) ), \\ M_\lambda^{\pm}, & z=\lambda \in (-\infty, -1) \cup (1, \infty), \end{cases}
\end{equation*}
are continuous from $\overline{\mathbb{C}_\pm} \setminus \{ -1, 1 \}$ to $\mathsf{B}(H^{-1/2}(\Gamma)^4)$.
    \item[$(v)$] The operator $M_z^2 - \frac{1}{4} I_4$ gives rise to a bounded operator
    \begin{equation*}
      M_z^2 - \frac{1}{4} I_4: H^{-1/2}(\Gamma)^4 \rightarrow H^{1/2}(\Gamma)^4.
    \end{equation*}
    \item[$(vi)$] The operator $\beta M_z + M_z \beta$ gives rise to a bounded operator
    \begin{equation*}
      \beta M_z + M_z \beta: H^{-1/2}(\Gamma)^4 \rightarrow H^{1/2}(\Gamma)^4.
    \end{equation*}
\end{itemize}
\end{proposition}
\begin{proof}
$(i)$ Let $S^*$ be given by~\eqref{def_S_star} and fix $z\in \mathsf{res}(  A_{0})$. From Proposition~\ref{Proposition_G_Z_map}~$(v)$ we obtain 
$\ker G_z^*=\ran(S-\Bar z)$. Moreover, as $\ran G_z^*=H^{1/2}(\Gamma)^4$ is closed also $\ran G_z$ is closed and hence 
\begin{equation*}
  \ran G_z=\bigl(\ker G_z^*\bigr)^\bot= \ker (S^*-z).
\end{equation*}
Since the graph norm of $S^*$ and the $L^2(\mathbb R^3)^4$ norm are equivalent on $\ker (S^*-z)$ we conclude
\begin{equation*}
 G_z\in\mathsf{B}\bigl(H^{-1/2}(\Gamma)^4, \dom  S^*\bigr),
\end{equation*}
when $\dom  S^*$ is equipped the graph norm of $S^*$. With the extension of the trace operator $\gamma_0 = \frac{1}{2} (\gamma_0^+ + \gamma_0^-)$ from Lemma~\ref{Lemma_trace} the claim of $(i)$ follows.
 
$(ii)$-$(iv)$ follow directly from Proposition~\ref{Proposition_G_Z_map}.

$(v)$ is shown in \cite[Proposition~4.4]{BeHo}, see also Remark~\ref{remark_gamma_tilde} below.

$(vi)$ follows from the discussion before \cite[Proposition~2.1]{BEHL19} and \cite[Theorem~6.11]{McLe}.
\end{proof}

\begin{remark} \label{remark_gamma_tilde}
  It is worth to mention that the operators $G_z$ and $M_z$ defined by~\eqref{gz} coincide with the maps $\widetilde{\gamma}(z)$ and $\widetilde{M}(z)$ 
  introduced in \cite[Proposition~4.4]{BeHo}. In fact, for $G_z$ and $\widetilde{\gamma}(z)$ this follows as their duals coincide;
  for $M_z$ and $\widetilde{M}(z)$ this follows from their definitions in~\eqref{gz} and \cite[Proposition~4.4]{BeHo}.
\end{remark}

\section{Dirac operators with electrostatic and Lorentz scalar $\delta$-shell
interactions} \label{section_delta_op}

In this section we recall the definition and some of the basic properties of
Dirac operators which are coupled with a combination of electrostatic and
Lorentz scalar $\delta$-shell potentials, as they were treated, e.g., in
\cite{AMV15, BEHL17, BEHL19}. Let $\nu$ be the unit normal vector field at
$\Gamma$ pointing outwards of $\Omega_{+}$. We define for $\eta, \tau \in \mathbb{R}$ the operator
\begin{equation}%
\begin{split}
    A_{\eta,\tau}f &:= (-i\alpha\cdot
\nabla+\beta)f_{+}\oplus(-i\alpha\cdot\nabla+\beta)f_{-}, \\
  \dom   A_{\eta,\tau} &:= \big\{
f=f_{+}\oplus f_{-}\in \dom   S^*: \\
  & \qquad \qquad  
-i(\alpha\cdot\nu)(\gamma_0^+ f-\gamma_0^-f%
)=\tfrac{1}{2} (\eta I_{4}+\tau\beta)(\gamma_0^+ f+ \gamma_0^-f)\big\},
\end{split}
\label{def_delta_op}%
\end{equation}
with $S^*$ in~\eqref{def_S_star} and $\gamma_0^\pm$ denotes the trace operator from Lemma~\ref{Lemma_trace}.
In the next proposition we recall in the case of non-critical interaction
strengths $\eta^{2}-\tau^{2}\neq4$ the qualitative spectral properties and a resolvent formula for the operator
$A_{\eta,\tau}$; cf. \cite[Lemma~3.3, Theorem~3.4, and Theorem~4.1]{BEHL19} or \cite[Theorem~4.4]{BEHL17}. 
We do not discuss the case of critical interaction strengths 
$\eta^2 - \tau^2 = 4$ here. In this situation the spectral properties of  $A_{\eta, \tau}$ are different from the non-critical case; cf. \cite{BeHo, OV17}.

\begin{proposition}
\label{proposition_basic_delta_op} Let $\eta,\tau
\in\mathbb{R}$ such that $\eta^{2}-\tau^{2}\neq4$ and let $G_{z}$ and $M_{z}$ be defined as in
\eqref{gz}. Then the operator $A_{\eta,\tau}$ in \eqref{def_delta_op} is self-adjoint in $L^2(\mathbb{R}^3)^4$ and the following is true.

\begin{itemize}
\item[$(i)$] $\sigma_{\textup{ess}}(A_{\eta,\tau%
})=\sigma_{\textup{ess}}(A_{0})=(-\infty,-1]\cup\lbrack1,\infty)$.

\item[$(ii)$] $z\in\sigma_{\textup{disc}}(A_{\eta,\tau})$ if and only if $-1\in\sigma((\eta I_{4}%
+\tau\beta)M_{z})$.

\item[$(iii)$] For $z\in\mathsf{res}(A_{\eta,\tau})$ the operator $I_4+(\eta I_{4}+\tau\beta)M_{z}$ is boundedly invertible in 
$H^{-1/2}(\Gamma)^4$ and with
\begin{equation} \label{def_Lambda}
  \Lambda^{\eta, \tau}_z := \left( I_4+(\eta I_{4}+\tau\beta)M_{z}\right)^{-1} (\eta I_{4}+\tau\beta) \in \mathsf{B}\big(  H^{-1/2}(\Gamma) \big)
\end{equation}
one has the resolvent formula
\begin{equation} \label{krein_resolvent_formula}
R_{z}^{\eta,\tau}:=(A_{\eta%
,\tau}-z)^{-1}=R_{z}^{0}-G_{z} \Lambda_z^{\eta, \tau} G_{\Bar z}^*.
\end{equation}

\item[$(iv)$] The discrete spectrum of $A_{\eta,\tau%
}$ in $(-1,1)$ is finite.
\end{itemize}
\end{proposition}

\begin{remark} \label{remark_regularity}
  One can show that a generic function in $\dom   S^*$ does not possess any positive Sobolev regularity near $\Gamma$. 
  However, in the non-critical case $\eta^{2}-\tau^{2}\neq4$ it was shown 
  in \cite[Theorem~3.4]{BEHL19} that $\dom   A_{\eta, \tau} \subset H^1(\Omega_+)^4 \oplus H^1(\Omega_-)^4$.
\end{remark}

In the following proposition we discuss the existence of embedded eigenvalues.

\begin{proposition}
\label{proposition_embedded_eigenvalues} Let $\eta,\tau \in\mathbb{R}$ such that $\eta^{2}-\tau^{2}\neq4$, let $A_{\eta,\tau}$ be 
defined by \eqref{def_delta_op}, and assume that $\Omega_+$ is connected.

\begin{itemize}
\item[$(i)$] If $\eta^{2}-\tau^{2}\neq\pm4$, then
$A_{\eta,\tau}$ has no embedded eigenvalues in
$(-\infty,-1)\cup(1,\infty)$.

\item[$(ii)$] If $\eta^{2}-\tau^{2}=-4$, then
$A_{\eta,\tau}$ has a discrete set of embedded
eigenvalues in $(-\infty,-1)\cup(1,\infty)$ which may only accumulate at
$\pm\infty$.
\end{itemize}
\end{proposition}

\begin{proof}
  Assertion $(i)$ can be shown in the same way as \cite[Theorem~3.7]{AMV15}; cf. the 
  discussion after this result. To get the result from item~$(ii)$ we note first that for $\eta^{2}-\tau^{2}=-4$ one has the decoupling 
  \begin{equation*}
    A_{\eta, \tau} = B_{\eta, \tau}(\Omega_+) \oplus B_{\eta, \tau}(\Omega_-),
  \end{equation*}
  where $B_{\eta, \tau}(\Omega_\pm)$ is a self-adjoint Dirac operator acting in $L^2(\Omega_\pm)^4$ with suitable boundary conditions on 
  $\Gamma$; cf. \cite[Lemma~3.1]{BEHL19}. One can show in the same way as in \cite[Theorem~3.7]{AMV15} that 
  $B_{\eta, \tau}(\Omega_+)$ has no eigenvalues in $(-\infty,-1)\cup(1,\infty)$. On the other hand, according to Remark~\ref{remark_regularity} the domain 
  of definition of $B_{\eta, \tau}(\Omega_-)$ is contained in $H^1(\Omega_-)^4$, which implies that
  the resolvent of $B_{\eta, \tau}(\Omega_-)$ is compact. Hence $B_{\eta, \tau}(\Omega_-)$ and thus also $A_{\eta, \tau}$ 
  have a discrete set of eigenvalues in $(-\infty,-1)\cup(1,\infty)$ possibly accumulating at $\pm \infty$.
\end{proof}

The map $\Lambda_z^{\eta, \tau}$ appearing in the Kre\u{\i}n type resolvent formula in Proposition~\ref{proposition_basic_delta_op} will be important for our later analysis.
In the following proposition we discuss some basic properties of $\Lambda_z^{\eta, \tau}$; in particular, we extend the limiting absorption principle for $M_z$ from Proposition~\ref{Proposition_m_Z_map} to $\Lambda_z^{\eta, \tau}$. This will be a key ingredient to show the limiting absorption principle for $A_{\eta, \tau}$ in Theorem~\ref{theorem_LAP}.

\begin{proposition} \label{proposition_LAP_Birman_Schwinger}
Let $\eta,\tau\in\mathbb{R}$ such that
$\eta^{2}-\tau^{2}\neq4$ and let $\Lambda_z^{\eta, \tau}$, $z \in \mathsf{res}(A_{\eta,\tau})$, be defined by~\eqref{def_Lambda}. Then the following assertions are true.

\begin{itemize}
\item[$(i)$] For $z_1, z_2 \in \mathsf{res}(A_{\eta,\tau})$ the relation
\begin{equation*}
  \Lambda_{z_1}^{\eta, \tau} - \Lambda_{z_2}^{\eta, \tau} = (z_2-z_1) \Lambda_{z_1}^{\eta, \tau} G_{\Bar z_1}^* G_{z_2} \Lambda_{z_2}^{\eta, \tau} 
\end{equation*}
holds.
\item[$(ii)$] Viewing $\Lambda_z^{\eta, \tau}$ as an operator in $\B(H^{1/2}(\Gamma)^4, H^{-1/2}(\Gamma)^4)$, one has $(\Lambda_z^{\eta, \tau})^* = \Lambda_{\Bar z}^{\eta, \tau}$.
\item[$(iii)$] The map
$z\mapsto \Lambda^{\eta, \tau}_z \in\mathsf{B}(H^{-1/2}(\Gamma)^4)$
is holomorphic 
on $\mathsf{res}\left(  A_{0}\right)  \setminus \sigma_\textup{disc}(A_{\eta, \tau})$.

\item[$(iv)$] There exists a closed set $\mathcal{N}_{\infty}\subset\left(
-\infty,-1\right)  \cup\left(  1,\infty\right)  $ with Lebesgue measure zero
such that the limits
\begin{equation}\label{jadoch}
  \Lambda_\lambda^{\eta, \tau, \pm} := \lim_{\varepsilon \searrow 0} \Lambda_{\lambda \pm i \varepsilon}^{\eta, \tau}
\end{equation}
exist in $\mathsf{B}(  H^{-1/2}(\Gamma)^4)  $ for all
$\lambda\in\left(  \left(  -\infty,-1\right)  \cup\left(  1,\infty\right)
\right)  \setminus\mathcal{N}_{\infty}$ and the maps
\begin{equation*}
  z \mapsto \Lambda_z^{\eta, \tau, \pm} := \begin{cases} \Lambda_z^{\eta, \tau}, & z \in \mathbb{C} 
  \setminus \bigl( (-\infty, -1] \cup  [1, \infty)\cup \sigma_\text{\rm disc}(A_{\eta, \tau})\bigr), \\ \Lambda_\lambda^{\eta, \tau, \pm}, & z=\lambda \in ((-\infty, -1) \cup (1, \infty) ) \setminus \mathcal{N}_\infty, \end{cases}
\end{equation*}
are continuous from $\overline{\mathbb{C}_\pm} \setminus (\sigma_\text{\rm disc}(A_{\eta, \tau}) \cup \mathcal{N}_\infty)$ to  $\mathsf{B}(H^{-1/2}(\Gamma)^4)$.
\end{itemize}
\end{proposition}

\begin{proof}
$(i)$ We introduce the notation $B := \eta I_4 + \tau \beta$. Using the resolvent identity and Proposition~\ref{Proposition_m_Z_map}~$(ii)$ we get 
\begin{equation*}
  \begin{split}
    \Lambda_{z_1}^{\eta, \tau} &- \Lambda_{z_2}^{\eta, \tau} = (I_4 + B M_{z_1})^{-1} B - (I_4 + B M_{z_2})^{-1} B \\ 
    &= (I_4 + B M_{z_1})^{-1} B (M_{z_2} - M_{z_1}) (I_4 + B M_{z_2})^{-1} B 
    = (z_2 - z_1) \Lambda_{z_1}^{\eta, \tau} G_{\Bar z_1}^* G_{z_2} \Lambda_{z_2}^{\eta, \tau},
  \end{split}
\end{equation*}
which is the claimed result.

$(ii)$ follows from the fact that $\widehat{M}_z := M_z^* \in 
\B(H^{1/2}(\Gamma)^4)$ is given by the restriction $\widehat{M}_z = M_{\Bar z} 
\upharpoonright H^{1/2}(\Gamma)^4$; cf. 
\cite[Proposition~4.4~(ii)]{BeHo} and also 
Remark~\ref{remark_gamma_tilde}. In fact, together with 
Proposition~\ref{proposition_basic_delta_op}~($iii$) this implies that 
$I_4+(\eta I_{4}+\tau\beta) \widehat{M}_{z}$ is boundedly invertible in 
$H^{1/2}(\Gamma)^4$ and hence, we have for $\varphi, \psi \in 
H^{1/2}(\Gamma)^4$ that
\begin{equation*}
   \begin{split}
     &\big( \Lambda_z^{\eta, \tau} \varphi, \psi \big)_{H^{-1/2}(\Gamma)^4\times H^{1/2}(\Gamma)^4}\\
     & \qquad  = \big( ( I_4+(\eta I_{4}+\tau\beta)M_{z})^{-1} (\eta 
I_{4}+\tau\beta) \varphi, \psi \big)_{H^{-1/2}(\Gamma)^4\times H^{1/2}(\Gamma)^4} \\
       &\qquad =  \big( ( I_4+(\eta I_{4}+\tau\beta)\widehat{M}_{\Bar z})^{-1} (\eta 
I_{4}+\tau\beta) \varphi, \psi \big)_{H^{1/2}(\Gamma)^4\times H^{-1/2}(\Gamma)^4} \\
       &\qquad =  \big( \varphi, (\eta I_{4}+\tau\beta) \left( I_4+M_{\Bar z} 
(\eta I_{4}+\tau\beta)\right)^{-1}  \psi \big)_{H^{1/2}(\Gamma)^4\times H^{-1/2}(\Gamma)^4} \\
       &\qquad =  \big( \varphi,  ( I_4+(\eta I_{4}+\tau\beta)M_{\Bar z})^{-1} 
(\eta I_{4}+\tau\beta) \psi \big)_{H^{1/2}(\Gamma)^4\times H^{-1/2}(\Gamma)^4}\\
       &\qquad =  \big( \varphi,  \Lambda_{\Bar z}^{\eta, \tau} \psi \big)_{H^{1/2}(\Gamma)^4\times H^{-1/2}(\Gamma)^4},
   \end{split}
\end{equation*}
which is the claim of this item.

$(iii)$ First we show that the identity
\begin{equation}
\left(  I_{4}+(\eta I_{4}+\tau\beta)M_{z}\right)
\left(  I_{4}-(\eta I_{4}+\tau\beta)M_{z}\right)
=I_{4}-\frac{1}{4}(\eta^{2}-\tau^{2})I_{4}-K_{z}
\label{Weyl_id}%
\end{equation}
holds with a compact operator $K_{z}$ in $H^{-1/2}(\Gamma)^4$. To prove
\eqref{Weyl_id} we note that
\begin{equation*}
I_{4}-(\eta I_{4}+\tau\beta)M_{z}=I_{4}-M_{z}%
(\eta I_{4}-\tau\beta)+K_{1,z}%
\end{equation*}
with
\begin{equation*}
K_{1,z}:=M_{z}(\eta I_{4}-\tau\beta)-(\eta I_{4}+\tau\beta)M_{z}=-\tau(M_{z}\beta+\beta M_{z}).
\end{equation*}
Since $H^{1/2}(\Gamma)^4$ is compactly embedded in $H^{-1/2}(\Gamma)^4$ it follows from 
Proposition~\ref{Proposition_m_Z_map}~$(vi)$ that $K_{1,z}$ is a compact operator in
$H^{-1/2}(\Gamma)^4$. Hence, also
\begin{equation*}
K_{2,z}:=(I_{4}+(\eta I_{4}+\tau\beta)M_{z})K_{1,z}%
\end{equation*}
is compact in $H^{-1/2}(\Gamma)^4$. Thus, we have
\begin{equation*}
\begin{split}
 \left(  I_{4}+(\eta I_{4}+\tau\beta)M_{z}\right)&
\left(  I_{4}-(\eta I_{4}+\tau\beta)M_{z}\right) \\
  &=\left(  I_{4}+(\eta I_{4}+\tau\beta)M_{z}\right)
\left(  I_{4}-M_{z}(\eta I_{4}-\tau\beta)\right)
+K_{2,z}\\
&  =I_{4}-(\eta I_{4}+\tau\beta)(M_{z})^{2}%
(\eta I_{4}-\tau\beta)+\tau(\beta
M_{z}+M_{z}\beta)+K_{2,z}\\
&  =I_{4}-\frac{1}{4}(\eta^{2}-\tau^{2})I_{4}-K_{z}%
\end{split}
\end{equation*}
with
\begin{equation*}
K_{z}:=(\eta I_{4}+\tau\beta)\left(  (M_{z}
)^{2}-\frac{1}{4}\right)  (\eta I_{4}-\tau\beta
)-\tau(\beta M_{z}+M_{z}\beta)-K_{2,z},
\end{equation*}
which is compact in $H^{-1/2}(\Gamma)^4$ by Proposition~\ref{Proposition_m_Z_map}~$(v)$-$(vi)$.
This shows~\eqref{Weyl_id}.

By Proposition~\ref{proposition_basic_delta_op}~$(iii)$ and~\eqref{Weyl_id} we have for $z\in \mathsf{res}(A_{\eta,\tau})\cap \mathsf{res}(A_{-\eta,-\tau})$ that
\begin{equation*}
\left(  I_{4}-(\eta I_{4}+\tau\beta)M_{z}\right)
^{-1}\left(  I_{4}+(\eta I_{4}+\tau\beta
)M_{z}\right)  ^{-1}=\left(  I_{4}-\frac{1}{4}(\eta^{2}
-\tau^{2})I_{4}-K_{z}\right)  ^{-1} 
\end{equation*}
belongs to $\mathsf{B}(H^{-1/2}(\Gamma)^4)$.
Moreover, the map $I_{4}-1/ 4(\eta^{2}-\tau^{2})I_{4}-K_{z}$ is holomorphic 
in $\mathsf{res}(A_0)$ 
due to the holomorphy
 of $M_z$ shown in Proposition~\ref{Proposition_m_Z_map}~$(iii)$ and $K_{z}$ is compact in $H^{-1/2}(\Gamma)^4$.
Therefore, the analytic 
Fredholm theorem \cite[Theorem VI.14]{ReSi} implies that
\begin{equation*}
  z \mapsto \left(I_{4}-\frac{1}{4}(\eta^{2}-\tau^{2})I_{4}-K_{z}\right)^{-1}
\end{equation*}
is holomorphic 
in $\B(H^{-1/2}(\Gamma)^4)$ for 
$z \in \mathsf{res}\left(  A_{0}\right) \setminus (\sigma_\text{disc}(A_{\eta, \tau}) \cup \sigma_\text{disc}(A_{-\eta, -\tau}) \cup \mathcal{N}_0)$, 
where $\mathcal{N}_0$ is a discrete set in $ \mathsf{res}(A_0)$. 
Since $z\mapsto M_{z}\in\mathsf{B}(H^{-1/2}(\Gamma)^4)$ is holomorphic 
on
$\mathsf{res}\left(  A_{0}\right)  $ by Proposition~\ref{Proposition_m_Z_map}, we conclude from
\begin{equation}
\Lambda_z^{\eta, \tau}=\left(  I_{4}-(\eta I_{4}+\tau\beta
)M_{z}\right)  \left(  I_{4}-\frac{1}{4}(\eta^{2}-\tau^{2})I_{4}-K_{z}\right)^{-1} (\eta I_4 + \tau \beta) \label{Weyl_id_2}%
\end{equation}
that $z\mapsto \Lambda_z^{\eta, \tau} \in\mathsf{B}(H^{-1/2}(\Gamma)^4)$ is holomorphic 
on
$\mathsf{res}(  A_0) \setminus (\sigma_\text{disc}(A_{\eta, \tau}) \cup \sigma_\text{disc}(A_{-\eta, -\tau}) \cup \mathcal{N}_0)$. 
Finally, by Proposition~\ref{proposition_basic_delta_op} and holomorphy
 this extends to all 
$z \in  \mathsf{res}\left(  A_{0}\right) \setminus \sigma_\text{disc}(A_{\eta, \tau})$.

$(iv)$ Note first that the limit properties of $z\mapsto M_{z}$ for
$z=\lambda\pm i\varepsilon$ and $\varepsilon \searrow 0$ with $\lambda\in\left(  -\infty,-1\right)
\cup\left(  1,\infty\right)  $ extend to
\begin{equation*}
z\mapsto T_{z}:=I_{4}-\frac{1}{4}(\eta^{2}-\tau^{2})I_{4}-K_{z}.
\end{equation*}
More precisely, it follows from Proposition~\ref{Proposition_m_Z_map} and \eqref{Weyl_id}
that
\begin{equation*}
T_{\lambda}^{\pm}:=\lim_{\varepsilon \searrow 0}T_{\lambda\pm i\varepsilon}\in\mathsf{B}%
\big(H^{-1/2}(\Gamma)^4\big),\qquad\lambda\in (-\infty,-1)\cup (1,\infty).
\end{equation*}
It is also clear from the considerations above that $T_{z}$ depends
analytically on $z\in\mathsf{res}\left(  A_{0}\right)  $, that $T_{z}$ has
a bounded inverse for $z\in\mathbb{C}\setminus\mathbb{R}$, and that $T_z$ can be extended to the mappings
\begin{equation*}
  z \mapsto T^\pm_z := \begin{cases} T_z, & z \in \mathbb{C} \setminus ((-\infty, -1] \cup [1, \infty)), \\ T_\lambda^\pm, & z=\lambda \in (-\infty, -1) \cup (1, \infty),  \end{cases}
\end{equation*}
which are continuous from $\overline{\mathbb{C}_\pm} \setminus \{ -1, 1 \}$ to $\mathsf{B}
(H^{-1/2}(\Gamma)^4)$, see Proposition~\ref{Proposition_m_Z_map}~$(iv)$. Therefore,
\cite[Theorem 9.10.2]{Sch} implies that there exists a set $\mathcal{N}%
_{\infty}\subset\left(  -\infty,-1\right)  \cup\left(  1,\infty
\right)  $ with Lebesgue measure zero such that $(T_{\lambda}^{\pm})^{-1}%
\in\mathsf{B}(H^{-1/2}(\Gamma)^4)$ for $\lambda\in\left(  -\infty,-1\right)
\cup\left(  1,\infty\right)  \setminus\mathcal{N}_{\infty}$. Next, let $\lambda \in ((-\infty, -1) \cup (1,\infty)) \setminus \mathcal{N}_\infty$ be fixed. Then we have for a small $\delta$
\begin{equation} \label{equation_T_z}
  T_{\lambda + \delta}^{\pm} = T_{\lambda}^{\pm} \big(I_4 + (T_{\lambda}^{\pm})^{-1} (T_{\lambda + \delta}^{\pm} - T_{\lambda}^{\pm})\big).
\end{equation}
With the continuity of $z \mapsto T^\pm_z$ and the Neumann formula we deduce from this that the set $\sigma_\text{disc}(A_{\eta, \tau}) \cup \sigma_\text{disc}(A_{-\eta, -\tau})\cup \mathcal{N}_0 \cup \mathcal{N}_\infty$, on which $T_z^\pm$ is not invertible, is closed. With a similar consideration as in~\eqref{equation_T_z} with $\delta \in \overline{\mathbb{C}_\pm}$ we find that
\begin{equation} \label{T_pm_inv_cont}
  z \mapsto (T^\pm_z)^{-1} = \begin{cases} (T_z)^{-1}, & z \in \mathsf{res}\left(  A_{0}\right) \setminus (\sigma_\text{disc}(A_{\eta, \tau}) \cup \sigma_\text{disc}(A_{-\eta, -\tau}) \cup \mathcal{N}_0), \\ (T_\lambda^\pm)^{-1}, & z=\lambda \in ( (-\infty, -1) \cup (1, \infty) ) \setminus \mathcal{N}_\infty,  \end{cases}
\end{equation}
is continuous in $\mathsf{B}(  H^{-1/2}(\Gamma)^4)  $.

Now, it is clear from the above considerations that $\mathcal{N}_\infty$ is closed and with the help of item~($iii$), \eqref{Weyl_id_2}, \eqref{T_pm_inv_cont}, 
and Proposition~\ref{Proposition_m_Z_map}~($iv$) we find that $\Lambda_z^{\eta, \tau, \pm}$ in~\eqref{jadoch} is continuous from 
$\overline{\mathbb{C}_\pm} \setminus (\sigma_\text{disc}(A_{\eta, \tau}) \cup \mathcal{N}_\infty)$ to $\mathsf{B}(  H^{-1/2}(\Gamma)^4)  $. 
This finishes the proof of this proposition.
\end{proof}


\begin{remark} \label{remark_singular_perturbation}
  By Proposition~\ref{proposition_LAP_Birman_Schwinger}~$(i)$-$(ii)$ the map $\Lambda_z^{\eta, \tau}$ defined in~\eqref{def_Lambda} satisfies the relations~\textup{(2.6)} and~\textup{(2.7)} in \cite{MaPo} and so, $A_{\eta, \tau}$ fits into the framework of \cite[Section~2]{MaPo}. In particular,~\eqref{krein_resolvent_formula} corresponds to formula~\textup{(2.10)} in \cite{MaPo}; note that the resolvents in \cite{MaPo} have a different sign than in this paper.
\end{remark}

Combining the Kre\u{\i}n type resolvent formula from Proposition~\ref{proposition_basic_delta_op} with Proposition~\ref{Proposition_G_Z_map} and 
Proposition~\ref{proposition_LAP_Birman_Schwinger} we get the limiting absorption principle for $A_{\eta, \tau}$.

\begin{theorem} \label{theorem_LAP}
  Let $\eta,\tau \in\mathbb{R}$ such that $\eta^{2}-\tau^{2} \neq 4$ and let $A_{\eta,\tau}$ be defined by
\eqref{def_delta_op}. Then there exists a closed set $\mathcal{N} \subset \mathbb{R} \setminus \{ -1, 1 \}$ with Lebesgue measure zero
such that for all $\lambda\in \mathbb{R} \setminus (\mathcal{N} \cup \{-1, 1 \})$ and $w>\frac{1}{2}$ the limits 
\begin{equation*}
  R^{\eta ,\tau  ,\pm}_{\lambda} := \lim_{\varepsilon\searrow 0} \big(A_{\eta, \tau} -(\lambda \pm i\varepsilon) \big)^{-1}
\end{equation*}
exist in the topology of $\B(L^{2}_{w}(\RE^{3})^4,L^{2}_{-w}(\RE^{3})^4)$, and they are explicitly given by
\begin{equation*}
  R^{\eta ,\tau  ,\pm}_{\lambda} = R^{0,\pm}_{\lambda} -
G^{\pm}_{\lambda} \Lambda_\lambda^{\eta, \tau, \pm} (G^{\mp}_{\lambda})^*,
\end{equation*}
where $G^{\pm}_{\lambda}$, $(G^{\mp}_{\lambda})^*$,  and $\Lambda_\lambda^{\eta, \tau,\pm}$ are defined as in~\eqref{limgz123}, \eqref{lim_G_star}, 
and~\eqref{jadoch}, respectively.
\end{theorem}

\begin{proof}
Recall first that $R^{0,\pm}_{\lambda}\in\mathsf{B}(L_{w}^{2}(\mathbb{R}^{3})^4, H_{-w}^{1}(\mathbb{R}^{3})^4 )$ by 
Proposition~\ref{Proposition_Res_D}~$(ii)$ for $s=0$ and $\lambda\in \mathbb{R} \setminus \{-1, 1 \}$, and hence, in particular, 
$R^{0,\pm}_{\lambda}\in\mathsf{B}(L_{w}^{2}(\mathbb{R}^{3})^4, L_{-w}^2(\mathbb{R}^{3})^4 )$ for $\lambda\in \mathbb{R} \setminus \{-1, 1 \}$. Next, we have
$(G^{\mp}_{\lambda})^*\in\mathsf{B}(L_{w}^{2}(\mathbb{R}^{3})^4, H^{1/2}(\Gamma)^4 )$ for $\lambda\in \mathbb{R} \setminus \{-1, 1 \}$ by 
Proposition~\ref{Proposition_G_Z_map}~$(vi)$ and hence also 
$(G^{\mp}_{\lambda})^*\in\mathsf{B}(L_{w}^{2}(\mathbb{R}^{3})^4, H^{-1/2}(\Gamma)^4 )$. 
Since $\Lambda_\lambda^{\eta, \tau, \pm}\in\mathsf{B}(H^{-1/2}(\Gamma)^4 )$ 
for $\lambda\in\mathbb{R} \setminus (\sigma_\textup{disc}(A_{\eta, \tau})\cup \mathcal{N}_\infty \cup \{-1, 1 \})$
by Proposition~\ref{proposition_LAP_Birman_Schwinger}~$(iv)$ 
and $G^{\pm}_{\lambda}\in\mathsf{B}(H^{-1/2}(\Gamma)^4,L_{-w}^{2}(\mathbb{R}^{3})^4)$ for $\lambda\in \mathbb{R} \setminus \{-1, 1 \}$ 
by Proposition~\ref{Proposition_G_Z_map}~$(iii)$
the assertion follows with the closed set $\mathcal{N}=\sigma_\textup{disc}(A_{\eta, \tau})\cup \mathcal{N}_\infty$; note that 
$\sigma_\textup{disc}(A_{\eta, \tau})$ is finite by Proposition~\ref{proposition_basic_delta_op}~$(iv)$.
\end{proof}

\begin{section}{The Scattering Matrix} \label{section_scattering}

In this section we calculate the scattering matrix for the couple $( A_{\eta, \tau}, A_0 )$, where
 $\eta, \tau \in \mathbb{R}$ are fixed such that $\eta^2  - \tau^2 \neq 4$ 
and $A_{\eta, \tau}$ is defined by~\eqref{def_delta_op}. First, we show the existence and 
completeness of the wave operators. We remark that their existence and completeness for smooth surfaces 
$\Gamma$ is shown in~\cite[Proposition~4.7]{BEHL19}, but we give a proof which also holds for $C^2$-surfaces $\Gamma$.

\begin{theorem}\label{W1} 
The scattering couple $( A_{\ee ,\es  },A_{0})$ is complete, that is, the strong limits 
\begin{equation*}
\begin{split}
W_{\pm}(A_{\ee ,\es  },A_{0})
:=&\,\text{s-}\lim_{t\to\pm\infty}e^{itA_{\ee ,\es  }}e^{-itA_{0}},\\
W_{\pm}(A_{0},A_{\ee ,\es  })
:=&\,\text{s-}\lim_{t\to\pm\infty}e^{itA_{0}}e^{-itA_{\ee ,\es  }}P^{\ee ,\es  }_\textup{ac}\,,
\end{split}
\end{equation*}
exist everywhere in $L^{2}(\RE^{3})^4$, and
$$
\ran(W_{\pm}(A_ {\ee ,\es  },A_{0}))
=(L^{2}(\RE^{3})^4)^{\ee ,\es  }_\textup{ac}\,,\qquad\ran(W_{\pm}(A_{0},A_{\ee ,\es  }))=L^{2}(\RE^{3})^4\,,
$$
and $W_{\pm}(A_ {\ee ,\es  },A_{0})^{*}=W_{\pm}(A_{0},A_{\ee ,\es  })$ hold; 
here $P^{\ee ,\es  }_\textup{ac}$ denotes the orthogonal projector onto the absolutely continuous subspace 
$(L^{2}(\RE^{3})^4)^{\ee ,\es  }_\textup{ac}$ relative to $A_{\ee ,\es  }$. 
\end{theorem}

\begin{proof} 
Let $\mathcal{N} \subset \mathbb{R} \setminus \{ -1, 1 \}$ be as in Theorem~\ref{theorem_LAP} and let ${I} \subset \mathbb{R} \setminus (\mathcal{N} \cup \{-1, 1 \})$ be compact. Then, by Proposition~\ref{Proposition_G_Z_map}~$(iv)$ 
\begin{equation*}
\sup_{(\lambda,\varepsilon)\in I\times (0,1)}\sqrt\varepsilon\,\|G_{\lambda\pm i\varepsilon}\|_{H^{-1/2}(\Gamma)^4, L^{2}(\mathbb{R}^3)^4}<\infty
\end{equation*}
holds and the continuity of $\Lambda_z^{\eta, \tau, \pm}$ from Proposition~\ref{proposition_LAP_Birman_Schwinger}~$(iv)$ and the fact 
that $H^{1/2}(\Gamma)^4$ is continuously embedded in $H^{-1/2}(\Gamma)^4$ imply 
\begin{equation*}
\sup_{(\lambda,\varepsilon)\in I\times (0,1)}\|\Lambda^{\ee ,\es  }_{\lambda\pm i\varepsilon}\|_{H^{1/2}(\Gamma)^4, H^{-1/2}(\Gamma)^4}<\infty\,.
\end{equation*}
Hence, the existence and completeness of the wave operators follows from \cite[Theorem 2.8]{MaPo} and Remark~\ref{remark_singular_perturbation}.
\end{proof}

\begin{remark} 
  {\rm (i)} Whenever the set $\mathcal{N}_{\infty}$ in Proposition \ref{proposition_LAP_Birman_Schwinger} (iv) is discrete, 
  then, proceeding as in \cite[Theorem 6.1]{Ag}, the limiting absorption principle provided in Theorem \ref{theorem_LAP} implies 
  absence of singular continuous spectrum and hence asymptotic completeness for the scattering couple $( A_{\ee ,\es  },A_{0})$.\\
  \noindent
  {\rm (ii)} In the so-called confinement case $\eta^2 - \tau^2 = -4$ the $\delta$-potential is impenetrable, i.e. the operator $A_{\eta, \tau}$ 
  decouples in the form
  $A_{\eta, \tau} = B_{\eta, \tau}(\Omega_+) \oplus B_{\eta, \tau}(\Omega_-)$, where $B_{\eta, \tau}(\Omega_\pm)$ 
  are self-adjoint operators in $L^2(\Omega_\pm)^4$; cf. \cite[Section~5]{AMV15}, \cite[Lemma~3.1]{BEHL19}, or the proof of 
  Proposition~\ref{proposition_embedded_eigenvalues}. This orthogonal decoupling extends to the 
  corresponding semigroups $e^{\pm itA_{\ee ,\es  }}$ in the definition of the wave operators and the scattering operator $S_{\ee ,\es  }$ below;
  for related considerations on the semigroup associated to $A_{\eta, \tau}$ in the confinement case we also refer to \cite[Section~5]{AMV15}.
  Since $\Omega_-$ is a bounded $C^2$-domain the spectrum of $B_{\eta, \tau}(\Omega_-)$ is discrete and hence the absolutely continuous spectra
  of $A_{\eta, \tau}$ and $B_{\eta, \tau}(\Omega_+)$ coincide. Therefore, for the scattering process only the operator $B_{\eta, \tau}(\Omega_+)$ 
  in the exterior domain $\Omega_+$ is relevant.
\end{remark}

The above theorem allows to define the unitary scattering operator in $L^{2}(\RE^{3})^4$ by  
$$
S_{\ee ,\es  }:=W_{+}(A_ {\ee ,\es  },A_{0})^{*}W_{-}(A_ {\ee ,\es  },A_{0})\,.
$$
To construct the associated scattering matrix, we introduce for any $\lambda\in\RE$ with $|\lambda|>1$   
\begin{equation*}
\begin{split}
L_{(\lambda)}^{2}({\mathbb S}^{2})^4:=\Bigg\{&
\psi_{\lambda}\in L^{2}({\mathbb S}^{2})^4:\\
&\frac12\left(I_{4}+\frac{\sqrt{\lambda^{2}-1}\,\alpha\cdot\xi+\beta}{\lambda}\right)\psi_{\lambda}(\xi)=\psi_{\lambda}(\xi)\ \text{for a.e. $\xi\in{\mathbb S}^{2}$}
\Bigg\}
\end{split}
\end{equation*}
and for any $w>1/2$, 
\begin{equation} \label{def_F_0}
F_{0}:L_{w}^{2}(\RE^{3})^4\to \int_{|\lambda|>1}^{\oplus} L_{(\lambda)}^{2}({\mathbb S}^{2})^4\,\sqrt{\lambda^{2}-1}\,|\lambda|\,d\lambda\,,\qquad 
F_{0}f(\lambda)=\widetilde f_{\lambda}\,,
\end{equation}
where the function $\widetilde f_{\lambda}\in L_{(\lambda)}^{2}({\mathbb S}^{2})^4$ is defined by 
$$
\widetilde        f_{\lambda}(\xi):=\frac12\left(I_{4}+\frac{\sqrt{\lambda^{2}-1}\,\alpha\cdot\xi+\beta}{\lambda}\right)
\widehat f(\sqrt{\lambda^{2}-1}\,\xi),
$$
and $\widehat f$ denotes the Fourier transform of $f$. 
Note that $\widetilde{f}_\lambda$ is well-defined, as $\widehat{f} \in H^w(\mathbb{R}^3)^4$  for $f \in L^2_w(\mathbb{R}^3)^4$ and hence, since $w>\frac{1}{2}$, $\widehat{f}$ has a trace on $\sqrt{\lambda^{2}-1} \mathbb{S}^2$.
The map $F_{0}$ extends to a unitary map on $L^{2}(\RE^{3})^4$, denoted by the same symbol,  
which diagonalizes $A_{0}$, i.e., $(F_{0}A_{0}f)(\lambda)=\lambda\, \widetilde f_{\lambda}$; see, e.g., \cite[Section 3.2]{Isozaki}. 
Then the scattering matrix is defined by
$$
S_{\ee ,\es  }(\lambda):
L_{(\lambda)}^{2}({\mathbb S}^{2})^4\to L_{(\lambda)}^{2}({\mathbb S}^{2})^4,\quad S_{\ee ,\es  }(\lambda)\widetilde f_{\lambda}
=(F_{0}S_{\ee ,\es  }f)(\lambda).
$$
In order to compute the scattering operator $S_{\eta, \tau}$ and the associated scattering matrix $S_{\ee ,\es  }(\lambda)$ we use 
the Birman-Kato invariance principle 
$$
W_{\pm}(A_{\ee,\es  },A_{0})=W_{\pm}(-R^{\ee,\es  }_{\mu},-R^{0}_{\mu})
$$
for some fixed
$\mu\in (-1,1)\cap\rho(A_{\ee ,\es  })$,
and so, by defining 
$$S_{\ee,\es  }^{\mu} := W_{+}(-R^{\ee,\es  }_{\mu}, -R^{0}_{\mu})^* W_{-}(-R^{\ee,\es  }_{\mu}, -R^{0}_{\mu})$$ 
we have
\begin{equation} \label{invariance_principle}
S_{\ee,\es  }=S_{\ee,\es  }^{\mu}\,.
\end{equation}
Below, we prove that all these objects associated to the scattering pair $( -R^{\ee,\es  }_{\mu}, -R^{0}_{\mu} )$ exist.
We note again, that the resolvents in \cite{MaPo} have a different sign as in this paper. Following the strategy developed 
\cite[Section 4]{MaPo}, we use the Birman-Yafaev  stationary scattering theory from \cite{Y} to provide the scattering matrix 
for the scattering couple $(-R^{\ee ,\es  }_{\mu}, -R^{0}_{\mu})$.

In the following let $\mu\in (-1,1)\cap\rho(A_{\ee ,\es  })$ be fixed. One verifies that
the unitary operator $F_{0}^{\mu}$ which diagonalizes $-R^{0}_{\mu}$ is 
\begin{equation*} 
  F^{\mu}_{0}f(\lambda)=\widetilde        f_{\lambda}^{\mu}:=\frac1\lambda\, \widetilde        f_{\mu-1/\lambda}\,,\quad\lambda\not=0\,,\quad\left|\mu-\frac1\lambda\right|>1.
\end{equation*} 
In the next preparatory lemma we compute the scattering matrix for the scattering couple $( -R^{\eta, \tau}_\mu, -R_\mu^0)$.

\begin{lemma}\label{W2}
The strong limits 
\begin{equation*}\label{WR}
W_{\pm}(-R^{\ee,\es}_{\mu},-R^{0}_{\mu})
:=\text{\rm s-}\lim_{t\to\pm\infty}e^{-itR^{\ee,\es}_{\mu}}e^{itR^{0}_{\mu}}
\end{equation*}
exist everywhere in $L^{2}(\RE^{3})^4$. 
Moreover, for any  $\lambda\not=0$ such that $\mu-\frac1\lambda\in \mathbb{R} \setminus (\mathcal{N} \cup [-1, 1 ])$, the scattering matrix $S^{\mu}_{\ee ,\es  }(\lambda)$ for the pair $( -R^{\eta, \tau}_\mu, -R_\mu^0)$ is given by
\begin{equation}\label{S1}
S^{\mu}_{\ee ,\es  }(\lambda)=I_4-2\pi i\,L^{\mu}_{\lambda}
\Lambda^{\ee,\es  } _{\mu}\big(I_4 - G^{*}_{\mu}\left(-R_{\mu}^{\ee,\es  }-(\lambda+ i0)\right)^{-1}G_{\mu}\Lambda^{\ee,\es  }  _{\mu}\big)
(L^{\mu}_{\lambda})^{*}\,,
\end{equation}
where
\begin{equation*}
L^{\mu}_\lambda: H^{-1/2}(\Gamma)^4\to L_{(\mu-1/\lambda)}^{2}({\mathbb S}^{2})^4\,,\quad 
L^{\mu}_{\lambda}\phi:=\frac1\lambda\,(F_{0}G_{\mu}\phi)(\mu-1/\lambda)\,.
\end{equation*}
\end{lemma}

\begin{proof}
We follow the same arguments as in the proof of \cite[Theorem 4.1]{MaPo}.
By $-R_{\mu}^{\ee,\es  }=-R_{\mu}^{0}+G_{\mu}\Lambda^{\ee,\es } _{\mu}G^{*}_{\mu}$, we can use \cite[Theorem 4', page 178]{Y}; for that, 
we notice that the maps denoted there by $G$ and $\mathcal{V}$ correspond to our $G_{\mu}^{*}$ and $\Lambda^{\ee,\es } _{\mu}$, respectively
\footnote{In fact, in the assumptions in \cite[Theorem 4', page 178]{Y} one has $\mathcal V=\mathcal V^*$ in the same Hilbert space $\mathcal G$.
However, one verifies that also more general perturbations of the form $G^*\mathcal V G$ with $\mathcal V:\mathcal G_1\rightarrow \mathcal G_{-1}$
in a rigging $\mathcal G_1\subset\mathcal G\subset\mathcal G_{-1}$ can be treated.},
and that $(\Lambda_\mu^{\eta, \tau})^* = \Lambda_\mu^{\eta, \tau}$ for our choice of a real $\mu$, when $\Lambda_\mu^{\eta, \tau}$ is viewed 
as an operator in $\B(H^{1/2}(\Gamma)^4, H^{-1/2}(\Gamma)^4)$, see Proposition~\ref{proposition_LAP_Birman_Schwinger}~$(ii)$. Moreover, the maps $B(z)$ and $Z_0(\lambda; G)$ appearing in \cite[Theorem 4', page 178]{Y} are in our situation
$B(z) = G_\mu^* (-R^{\eta, \tau}_\mu - z)^{-1} G_\mu$ and $Z_0(\lambda; G) \phi = (F_0^\mu G_\mu \phi)(\lambda) = L_\lambda^\mu \phi$, $\phi \in H^{-1/2}(\Gamma)^4$.

Let us check that the assumptions required in \cite[Theorem~4', page~178]{Y} are satisfied. First, since $G^{*}_{\mu}\in \B(L^{2}(\RE^{3})^4,H^{1/2}(\Gamma)^4)$, the operator $G^{*}_{\mu}$ is $|R^{0}_{\mu}|^{1/2}$-bounded. To proceed, we note that the relations (which follow from the resolvent identity)
\begin{equation}\label{RR}
\left(-R_{\mu}^{0}-z\right)^{-1}=-\frac1z\left(I_4-\frac1{z}\,R^{0}_{\mu-\frac1z}\right),\quad   \left(-R_{\mu}^{\ee ,\es  }-z\right)^{-1}=-\frac1z\left(I_4-\frac1{z}\,R^{\ee ,\es  }_{\mu-\frac1z}\right),
\end{equation}
and the limiting absorption principles for $A_{0}$ and $A_{\ee ,\es  }$ (see Proposition~\ref{Proposition_Res_D} and Theorem~\ref{theorem_LAP}) imply that the limits 
\begin{equation*}
\left(-R_{\mu}^{0}-(\lambda\pm i0)\right)^{-1}:=\lim_{\varepsilon\searrow 0}\left(-R_{\mu}^{0}-(\lambda\pm i\varepsilon)\right)^{-1}
\end{equation*}
for $\lambda\not=0, \mu-\frac1\lambda\not=\pm 1$, and
\begin{equation*}
\left(-R_{\mu}^{\ee ,\es  }-(\lambda\pm i0)\right)^{-1}:=
\lim_{\varepsilon\searrow 0}\left(-R_{\mu}^{\ee ,\es  }-(\lambda\pm i\varepsilon)\right)^{-1}
\end{equation*}
for $\lambda\not=0, \mu-\frac1\lambda\in\mathbb{R} \setminus (\mathcal{N} \cup \{-1, 1 \})$,
exist in $\B(L^{2}_{w}(\RE^{3}),L^{2}_{-w}(\RE^{3}))$.
Therefore, the limits 
\begin{equation*}
\lim_{\varepsilon\searrow 0}\, G^{*}_{\mu}(-R^{0}_{\mu}-(\lambda\pm i\varepsilon))^{-1}\,,
\quad
\lim_{\varepsilon\searrow 0}\, G^{*}_{\mu}(-R^{\ee,\es  }_{\mu}-(\lambda\pm i\varepsilon))^{-1}\,,
\end{equation*}
and 
\begin{equation*}
  \lim_{\varepsilon\searrow 0}\, G^{*}_{\mu}(-R^{\ee,\es  }_{\mu}-(\lambda\pm i\varepsilon))^{-1}G_{\mu}
\end{equation*}
exist. Thus, to get the claimed result we need to check the validity of the remaining assumption in 
\cite[Theorem 4', page 178]{Y}, namely that $G^{*}_{\mu}$ is weakly-$R^{0}_{\mu}$ smooth, i.e., by \cite[Lemma 2, page 154]{Y}, 
\begin{equation*}
\sup_{0<\varepsilon<1}\varepsilon\,\|G^{*}_{\mu} (-R^{0}_{\mu}-(\lambda\pm i\varepsilon))^{-1}\|_{L^{2}(\mathbb{R}^3)^4,H^{1/2}(\Gamma)^4}^{2}\le c_{\lambda}<\infty\,,\quad\text{a.e. $\lambda$}\,.
\end{equation*}
By \eqref{RR}, this is a consequence of
\begin{equation}\label{in2.2}
\sup_{0<\delta<1}\delta\,\|G^{*}_{\mu}R^{0}_{\mu-\frac1\lambda\pm i\delta}\|_{L^{2}(\mathbb{R}^3)^4,H^{1/2}(\Gamma)^4}^{2}\le C_{\lambda}<\infty\,,\quad\text{a.e. $\lambda$}\,.
\end{equation}
To show~\eqref{in2.2}, we compute for $z \in \mathbb{C} \setminus \mathbb{R}$
\begin{equation*}
\begin{split}
\|G^{*}_{\mu} R^{0}_{z}\|_{L^{2}(\mathbb{R}^3)^4,H^{1/2}(\Gamma)^4} 
 &=\|\gamma_{0} R^{0}_{\mu} R^{0}_{z}\|_{L^{2}(\mathbb{R}^3)^4,H^{1/2}(\Gamma)^4}\\
 &=\|\gamma_{0} R^{0}_{z} R^{0}_{\mu}\|_{L^{2}(\mathbb{R}^3)^4,H^{1/2}(\Gamma)^4}\\
&=
\|R^{0}_{\mu}(\gamma_{0} R^{0}_{z})^{*}\|_{H^{-1/2}(\Gamma)^4,L^{2}(\mathbb{R}^3)^4}\\
&\le
\|R^{0}_{\mu}\|_{L^{2}(\mathbb{R}^3)^4,L^{2}(\mathbb{R}^3)^4}\|G_{\bar z}\|_{H^{-1/2}(\Gamma)^4,L^{2}(\mathbb{R}^3)^4}\,.
\end{split}
\end{equation*}
With the help of~\eqref{sup_G_z} the last calculation shows that \eqref{in2.2} is indeed true. Thus, by \cite[Theorem 4', page 178]{Y}, 
the limits \eqref{WR} exist everywhere in $L^{2}(\RE^{3})^4$ and the corresponding scattering matrix is given by \eqref{S1}.
\end{proof}

With the invariance principle and Lemma~\ref{W2} it is now possible to compute the scattering matrix for the pair $( A_{\eta, \tau}, A_0 )$.

\begin{theorem}  \label{theorem_scattering_matrix}
The scattering matrix 
$$
S_{\ee ,\es  }(\lambda):
L_{(\lambda)}^{2}({\mathbb S}^{2})^4\to L_{(\lambda)}^{2}({\mathbb S}^{2})^4\,,\quad \lambda\in\mathbb{R}\setminus(\mathcal{N} \cup [-1,1])\,,
$$
for the scattering couple $(A_{\ee,\es}, A_{0})$ has the representation
\begin{equation}\label{S-matrix}
S_{\ee ,\es  }(\lambda)=I_4-2\pi iL_{\lambda} \Lambda^{\eta, \tau, +}_\lambda L_{\lambda}^{*}\,,\end{equation}
where  
$L_\lambda: H^{-1/2}(\Gamma)^4\to L_{(\lambda)}^{2}({\mathbb S}^{2})^4$ 
acts on any $\phi\in L^{2}(\Gamma)^4$ as  
$$
L_{\lambda}\phi(\xi):=
\frac12\left(I_{4}+\frac{\sqrt{\lambda^{2}-1}\,\alpha\cdot\xi+\beta}{\lambda}\right)\frac1{(2\pi)^{3/2}}\int_{\Gamma}e^{-i\sqrt{\lambda^{2}-1}\,\xi\cdot x}\phi(x)\,d\sigma(x)\,.
$$ 
\end{theorem}
\begin{proof}
Recall that by Theorem \ref {W1}, Lemma \ref{W2}, and by Birman-Kato invariance principle~\eqref{invariance_principle}, one has
$$
S_{\ee,\es  }=S_{\ee,\es  }^{\mu}\,.
$$
To get the representation in~\eqref{S-matrix}, we note first that $(F_0^\mu g)(\lambda)=\frac1\lambda (F_0 g)(\mu-\tfrac1\lambda)$ implies 
$(F_0 g)(\lambda) = (\mu- \lambda)^{-1} (F_0^\mu g)((\mu-\lambda)^{-1})$. Hence, we conclude with the invariance principle for any $f \in L^2(\mathbb{R}^3)^4$
\begin{equation*}
  \begin{split}
    S_{\ee ,\es  }(\lambda) \widetilde{f}_\lambda &= (F_0 S_{\ee ,\es  } f)(\lambda) = (F_0 S_{\ee ,\es  }^\mu f)(\lambda)\\
    &= (\mu - \lambda)^{-1} (F_0^\mu S_{\ee ,\es  }^\mu f)((\mu - \lambda)^{-1})
    = S_{\ee ,\es  }^\mu((\mu - \lambda)^{-1}) \widetilde{f}_\lambda,
  \end{split}
\end{equation*}
that means
(see also \cite[Equation (14), Section 6, Chapter 2]{Y})
\begin{equation}\label{SS}
S_{\ee ,\es  }(\lambda)=S^{\mu}_{\ee ,\es  }((-\lambda+\mu)^{-1})\,.
\end{equation}
Next, using~\eqref{RR}, Proposition~\ref{proposition_basic_delta_op}~$(iii)$, Proposition~\ref{Proposition_G_Z_map}~$(i)$, and Proposition~\ref{proposition_LAP_Birman_Schwinger}~$(i)$ we compute for $z \in \mathbb{C} \setminus \mathbb{R}$
\begin{equation*}
  \begin{split}
    \Lambda^{\ee,\es}_{\mu}\big(I_4 &- G^{*}_{\mu}\left(-R_{\mu}^{\ee,\es  }-z\right)^{-1}G_{\mu}\Lambda^{\ee,\es}_{\mu}\big)\\
    &= \Lambda^{\ee,\es}_{\mu}\left(I_4 + \frac{1}{z} G^{*}_{\mu}\left(I_4-\frac{1}{z} R_{\mu-\frac{1}{z}}^{\ee,\es  } \right) G_{\mu}\Lambda^{\ee,\es}_{\mu}\right) \\
    &= \Lambda^{\ee,\es}_{\mu}\left(I_4 + \frac{1}{z} G^{*}_{\mu} \left(G_{\mu} - \frac{1}{z} R_{\mu-\frac{1}{z}}^{0} G_{\mu} + \frac{1}{z} G_{\mu - \frac{1}{z}} \Lambda_{\mu - \frac{1}{z}}^{\eta, \tau} G_{\mu-\frac{1}{\bar z}}^* G_{\mu} \right) \Lambda^{\ee,\es}_{\mu}\right) \\
    &= \Lambda^{\ee,\es}_{\mu} + \frac{1}{z} \Lambda^{\ee,\es}_{\mu} G^{*}_{\mu} \left(G_{\mu - \frac{1}{z}} \Lambda^{\ee,\es}_{\mu}+ \frac{1}{z} G_{\mu - \frac{1}{z}} \Lambda_{\mu - \frac{1}{z}}^{\eta, \tau} G_{\mu-\frac{1}{\bar z}}^* G_\mu \Lambda^{\ee,\es}_{\mu} \right)  \\
    &= \Lambda^{\ee,\es}_{\mu} + \frac{1}{z} \Lambda^{\ee,\es}_{\mu} G^{*}_{\mu} \left(G_{\mu - \frac{1}{z}} \Lambda^{\ee,\es}_{\mu}+ G_{\mu - \frac{1}{z}} \bigl(\Lambda_{\mu - \frac{1}{z}}^{\eta, \tau} - \Lambda^{\ee,\es}_{\mu}\bigr) \right)  \\
    &= \Lambda^{\ee,\es}_{\mu} + \frac{1}{z} \Lambda^{\ee,\es}_{\mu} G^{*}_{\mu} G_{\mu - \frac{1}{z}} \Lambda_{\mu - \frac{1}{z}}^{\eta, \tau}  \\
    &=
    \Lambda^{\ee,\es } _{\mu-\frac1z}\,.
  \end{split}
\end{equation*}
Using this identity with $z=\lambda\pm i\varepsilon$ for $\mu-\frac1{\lambda}\in \mathbb{R} \setminus \mathcal{N}$ 
and considering the limit $\varepsilon\searrow 0$ we deduce with Proposition~\ref{proposition_LAP_Birman_Schwinger}, Theorem~\ref{theorem_LAP}, and \eqref{RR}
$$
\Lambda^{\ee,\es } _{\mu}\big(I_4-G^{*}_{\mu}\left(-R_{\mu}^{\ee,\es  }-(\lambda + i0)\right)^{-1}G_{\mu}\Lambda^{\ee,\es}  _{\mu}\big)
=\Lambda^{\ee,\es,+}_{\mu-\frac1{\lambda}}.
$$
Therefore, by Lemma \ref{W2} we have 
\begin{equation}\label{S111}
S^{\mu}_{\ee ,\es  }(\lambda)=I_4-2\pi i\,L^{\mu}_{\lambda}
\Lambda^{\ee,\es,+}_{\mu-\frac1{\lambda}}
(L^{\mu}_{\lambda})^{*}\,.
\end{equation}
Thus \eqref{S-matrix} follows from \eqref{S111}, equation~\eqref{SS}, and by setting  
$L_{\lambda}:= -L^{\mu}_{{(-\lambda+\mu)^{-1}}}$ (note that the minus sign does not change the final result, 
as $L_\lambda$ appears only in products with $L_\lambda^*$). Let us finally calculate the explicit action 
of $L_{\lambda}$ by using the definition of the map $F_{0}$ from~\eqref{def_F_0}. Since 
$(F_{0}R^{0}_{\mu} f)(\lambda)=(\lambda-\mu)^{-1}\widetilde f_\lambda$ and $G_\mu=R^{0}_{\mu} \gamma_0^{*}$ we have for $\phi \in L^2(\Gamma)^4$
\begin{align*}
L_{\lambda}\phi(\xi)&=(\lambda - \mu)((F_{0} R^{0}_{\mu} \gamma_0^{*}\phi)(\lambda))(\xi)\\
&=
\frac12\left(I_{4}+\frac{\sqrt{\lambda^{2}-1}\,\alpha\cdot\xi+\beta}{\lambda}\right)\widehat{\gamma_{0}^{*}\phi}(\sqrt{\lambda^{2}-1}\,\xi)\\
&=\frac12\left(I_{4}+\frac{\sqrt{\lambda^{2}-1}\,\alpha\cdot\xi+\beta}{\lambda}\right)\frac1{(2\pi)^{3/2}}\int_{\Gamma}e^{-i\sqrt{\lambda^{2}-1}\,\xi\cdot x}\phi(x)\,d\sigma(x)\,.
\end{align*}  
This completes the proof of Theorem~\ref{theorem_scattering_matrix}.
\end{proof}
\end{section}

\end{document}